\def\wgt{\mathop{\rm wgt}}

\def\rank{\mathop{\rm rank}}

\def\Prob{\mathop{\rm Prob}}

\def\Ftwo{\mathbb{F}_2}
\documentclass[aps,pra,reprint,
superscriptaddress,
longbibliography]{revtex4-1}
\usepackage{hyperref}
\usepackage{datetime}
\usepackage{amsmath}
\usepackage{amsthm}
\usepackage{thmtools,thm-restate}
\usepackage{amssymb}
\usepackage{mathrsfs}
\usepackage{graphicx}
\usepackage{color}

\graphicspath{{.}{figs/}}

\newtheorem{theorem}{Theorem}
\newtheorem{lemma}[theorem]{Lemma}
\newtheorem{statement}[theorem]{Statement}

\newtheorem{definition}{Definition}

\begin{document}

\title{Duality
  and free energy analyticity bounds for few-body Ising models\\ with
  extensive homology rank}

\author{Yi Jiang}
\affiliation{Department of Physics \& Astronomy,
  University of California, Riverside, California 92521, USA}

\author{Ilya Dumer}
\affiliation{Department of Electrical Engineering, University of
  California, Riverside, California 92521, USA}

\author{Alexey A. Kovalev}
\affiliation{Department of Physics \&
  Astronomy and Nebraska Center for Materials and Nanoscience,
  University of Nebraska, Lincoln, Nebraska 68588, USA}

\author{Leonid P. Pryadko}
\email{leonid.pryadko@ucr.edu}
\affiliation{Department of Physics \& Astronomy,
  University of California, Riverside, California 92521, USA}
\date{\today}

\begin{abstract}
  We consider pairs of few-body Ising models where each spin enters a
  bounded number of interaction terms (bonds), such that each model
  can be obtained from the dual of the other after freezing $k$ spins
  on large-degree sites.  Such a pair of Ising models can be
  interpreted as a two-chain complex with $k$ being the rank of the
  first homology group.  Our focus is on the case where $k$ is
  extensive, that is, scales linearly with the number of bonds $n$.
  Flipping any of these additional spins introduces a homologically
  non-trivial defect (generalized domain wall).  In the presence of
  bond disorder, we prove the existence of a low-temperature
  weak-disorder region where additional summation over the defects
  have no effect on the free energy density $f(T)$ in the
  thermodynamical limit, and of a high-temperature region where in the
  ferromagnetic case an extensive homological defect does not affect
  $f(T)$.  We also discuss the convergence of the high- and
  low-temperature series for the free energy density, prove the
  analyticity of limiting $f(T)$ at high and low temperatures, and
  construct inequalities for the critical point(s) where analyticity
  is lost.  As an application, we prove multiplicity of the
  conventionally defined critical points for Ising models on all
  $\{f,d\}$ tilings of the infinite hyperbolic plane, where
  $df/(d+f)>2$.  Namely, for these infinite graphs, we show that
  critical temperatures with free and wired boundary conditions
  differ, $T_c^{(\mathrm{f})}<T_c^{(\mathrm{w})}$.
\end{abstract}


\maketitle

\section{Introduction}
Singular behavior associated with a phase transition may emerge only
in the thermodynamical limit, as the system size goes to infinity.
One example are spin models on any finite-dimensional lattice, where
both the interaction strength and its range are finite.  Then the
thermodynamical limit is well defined thanks to the fact that boundary
contribution scales sublinearly with the system
size\cite{Griffiths-results-1972}.  Respectively, e.g., in the case of
an Ising model, the same transition can be alternatively defined as
the temperature where spontaneous magnetization appears, spin
susceptibility diverges, spin correlations start to decay
exponentially, domain wall tension is lost, or as a singular point of
the free energy\cite{Lebowitz-1977,Gruber-Lebowitz-1978,%
  Lebowitz-Pfister-1981,Griffiths-results-1972}.

Situation is different if we are interested in non-local models, where
the size of the boundary scales linearly with the size of the subset
induced by any finite set of vertices.  Examples include models on
infinite transitive expander graphs like a degree-regular
tree\cite{Lyons-1989,Lyons-review-2000} or regular $\{f,d\}$ tilings
of the hyperbolic plane, with $(f-2)(d-2)>2$.  Here, the bulk
quantities cannot be uniquely defined, and the transition temperature
may depend on both the quantity being probed and the boundary
conditions used to define the infinite-graph limit.  

Infinite hyperbolic tilings 
provide a natural short-scale regularization for a space with constant
negative curvature.  Interest in quantum field theory models on curved
space-time is motivated by quantum gravity and, in particular, the
AdS/CFT correspondence\cite{Birrel-Davies-book-1982,%
  Cognola-Kirsten-Zerbini-1993,Camporesi-1991,%
  Miele-Vitale-1997,Doyon-2003,Evenbly-Vidal-2011,%
  Hiroaki-Ishihara-Hashizume-2013}.  There is an independent interest
in models on curved spaces in statistical mechanics and condensed
matter communities, e.g., since curvature can serve as an additional
parameter to drive the criticality, or as a way to introduce
geometrical frustration in toy models of amorphous solids, supercooled
liquids, and metallic
glasses\cite{Nelson-1983,Tarjus-Kivelson-Nussinov-Viot-2005,%
  Vitelli-Lucks-Nelson-2006,Sausset-Tarjus-2007,%
  Giomi-Bowick-2007,Turner-Vitelli-Nelson-2010,%
  Garcia-etal-2015,Benedetti-2015}.  Models like percolation on more
general expander graphs and various random graph ensembles are also
common in network theory, e.g., such models occurred in relation to
internet stability and spread of infectious
diseases\cite{Cohen-Erez-benAvraham-Havlin-2000,%
  Albert-Barabasi-RMP-2002,Gai-Kapadia-2010,%
  Borner-Sanyal-Vspignani-ARIST-2007,%
  Sander-epidemics-2002,Costa-Oliveira-CorreaRocha-2011,%
  Danon-etal-2011}.  Finally, the strongest motivation to study
non-local Ising models comes from their
relation\cite{Dennis-Kitaev-Landahl-Preskill-2002,%
  Kovalev-Pryadko-SG-2015,Kovalev-Prabhakar-Dumer-Pryadko-2018} to
certain families of finite-rate quantum 
error-correcting codes (QECCs).  

In a companion paper\cite{Kovalev-Prabhakar-Dumer-Pryadko-2018}
devoted to error-correcting properties of QECCs, three of us studied
pairs of \emph{weakly-dual} few-body Ising models where each spin
enters a bounded number of interaction terms (bonds).  Each model can
be obtained from the exact dual of the other after freezing $k$ spins
which enter a large number of bonds.  For the related QECC, $k$ is the
number of encoded qubits, and its ratio to the number of bonds,
$R\equiv k/n$, is the code rate.  One can also map such a pair of
Ising models to a $2$-chain complex $\Sigma$, in which case $k$ is the
rank of the first homology group $H_1(\Sigma)$.  In particular, in
Ref.~\onlinecite{Kovalev-Prabhakar-Dumer-Pryadko-2018} we introduced
the \emph{homological difference} $\Delta F\ge0$, the difference of
the free energies of two models with and without the additional
summation over the homological defects, and gave the sufficient
conditions for the existence of a low-temperature low-disorder region
on the phase diagram where in the large-system limit $\Delta F=0$.

In the present work we study duality and phase transitions in general
Ising models with the help of the \emph{specific homological
  difference} scaled by the number of bonds, $\Delta f=\Delta F/n$,
focusing on the case where the homology rank $k$ scales linearly with
the number of bonds $n$.  Upon duality $\Delta f$ is mapped to
$R\ln 2-\Delta f^*$, where $\Delta f^*$ is the homological difference
for the other model in the pair, at the dual temperature.  Existence
of a low-temperature \emph{homological} region where asymptotically
$\Delta f=0$ implies that at high temperatures $\Delta f^*=R\ln2$;
with $R>0$ this implies the existence of at least two distinct points
where $\Delta f$ is non-analytic as a function of temperature.
Combining with the analysis of convergence of the high-temperature
series expansion for the free energy density, we obtain several bounds
for critical temperatures associated with the non-analyticity of
$\Delta f$ and limiting free energy densities of the two models.  Main
result is the inequality for the change of phase transition
temperature due to summation over the homological defects.  As an
application, we prove multiplicity of the conventionally defined
critical points for Ising models on all $\{f,d\}$ tilings of the
hyperbolic plane with $df/(d+f)>2$.  That is, we show that transition
temperatures with wired and free boundary conditions differ,
$T_c^{(\mathrm{w})}>T_c^{(\mathrm{f})}$, which extends the results of
Refs.\
\onlinecite{Wu-hyperb-2000,Schonmann-2001,Haggstrom-Jonasson-Lyons-2002}.

The paper is organized as follows.  We introduce the notations and
review some known facts from theory of general Ising models and theory
of QECCs in Sec.\ \ref{sec:background}.  Our results are given in
Sec.\ \ref{sec:results}, where we first discuss properties of the
homological difference $\Delta f$, analyze the convergence and
analyticity of free energy density for a sequence of weakly-dual Ising
model pairs, and finally apply the obtained results to Ising models on
$\{f,d\}$ tilings of the hyperbolic plane, additionally illustrating
the conclusions with numerical simulations.  We summarize the results
and list some related open questions in Sec.\ \ref{sec:discussion}.
Most of the proofs are given in the Appendices.

\section{Notations and background}
\label{sec:background}  
We  consider general
Ising models in Wegner's form\cite{Wegner-1971}, which describes joint
probability distribution of $r\equiv |{\cal V}|$ Ising spin variables,
$S_v\in\{-1,1\}$, associated with elements of the vertex set, ${\cal V}$,
\begin{equation}
  \label{eq:prob-distribution}
  \Prob\nolimits_\mathbf{e}[\{S\};\Theta;K, h]={1\over Z}\prod_{b\in\mathcal{B}}
  e^{K (-1)^{e_b} R_b}\prod_{v\in\mathcal{V}}e^{h S_v},
\end{equation}
where each bond $R_b\equiv \prod_{v\in{\cal V}} S_v^{\Theta_{vb}}$,
$b\in{\cal B}$, $|{\cal B}|=n$, is a product of the spin variables
corresponding to non-zero positions in the corresponding column of the
$r\times n$ binary coupling matrix $\Theta$, the binary ``error'' 
vector $\mathbf{e}$ with components $e_b$, $b\in{\cal B}$, describes
quenched disorder, and the dimensionless coupling coefficients are
$K\equiv \beta J$ and $h\equiv \beta h'$, where $J$ is the Ising
exchange constant, $h'$ is the magnetic field, and $\beta\equiv 1/T$
the inverse temperature in energy units.  The normalization
constant 
in Eq.~(\ref{eq:prob-distribution}) is the \emph{partition function},%
\begin{equation}
  \label{eq:Z}
Z\equiv  Z_{\mathbf{e}}(\Theta;K,  h)\equiv \sum_{\{
    S_v= \pm1\}}\prod_{b\in{\cal B}} 
  e^{K (-1)^{e_b} R_b}\prod_{v\in\mathcal{V}} e^{ h S_v}.  
\end{equation}
The partition function is commonly written in terms of the
corresponding logarithm, the \emph{free energy}, $F=-\ln Z$, or the
free energy density (per bond), $f=F/n$.

The binary coupling matrix $\Theta$ in
Eq.~(\ref{eq:prob-distribution}) can be interpreted geometrically in
terms of a bipartite \emph{Tanner} graph\cite{Tanner-graph-1981}, or,
equivalently, as the vertex-edge incidence matrix for a hypergraph
$\mathcal{H}=(\mathcal{V},\mathcal{B})$ with vertex set $\mathcal{V}$
and hyperedge (bond) set $\mathcal{B}$, with each hyperedge
$b\in\mathcal{B}$ a non-empty subset of the vertex set,
$b\subseteq \mathcal{V}$.  In comparison, in a (simple undirected)
graph $\mathcal{G}=(\mathcal{V},\mathcal{E})$, each edge
$b\in\mathcal{E}$ is an unordered pair of vertices,
$b=\{i,j\}\subseteq\mathcal{V}$.  The \emph{degree} $d_v$ of a vertex
$v\in\mathcal{V}$ in a (hyper)graph is the number of edges that
contain $v$, it is equal to the number of non-zero entries in the
$v$\,th row of the vertex-edge incidence matrix $\Theta$.  Similarly,
the size of an edge in a hypergraph is called its degree, $d_b=|b|$,
$b\in\mathcal{B}$.  In a graph, all edges are pairs of vertices, and
all columns of the incidence matrix $\Theta$ have exactly two non-zero
entries.

The probability distribution (\ref{eq:prob-distribution}) can be
characterized via the
corresponding marginals, the spin correlations
\begin{equation}
  \label{eq:marginal}
  \langle S_{\cal A}\rangle
  \equiv 
  \sum_{\{S_v=\pm1\}}S_{\cal A}\Prob\nolimits_\mathbf{e}(\{S\};\Theta;K,h), 
\end{equation}
where ${\cal A}\subseteq{\cal V}$ is a set of vertices,
$S_\mathcal{A}=\prod_{v\in\mathcal{A}}S_v$; by convention,
$ S_\emptyset=1$.  At $h=0$, on a finite system and with
$\mathbf{e}=\mathbf{0}$, non-zero expectation is obtained for the sets
(and only the sets) that can be constructed as products of
bonds\cite{Wegner-1971},
\begin{equation}
S_{\cal A}=\prod_{b:m_b\neq0} R_b=\prod_v \prod_b
S_v^{\Theta_{vb}m_b},\label{eq:bond-product}
\end{equation}
where bonds in the product correspond to non-zero positions
$m_b\neq0$ in the binary vector $\mathbf{m}\in\Ftwo^n$ of
\emph{magnetic charges}.  A number of \emph{correlation
  inequalities} for spin averages have been constructed, see, e.g.,
Refs.~\onlinecite{Percus-1975,Shlosman-1981}.
Particularly important for this work are Griffiths-Kelly-Sherman (GKS)
inequalities\cite{Griffiths-1967,Kelly-Sherman-1968},%
\begin{eqnarray}
  \label{eq:gks-ineq-one}
  \langle S_{\cal A}\rangle&\ge&0,\\
  \langle S_{\cal A}S_{\cal B}\rangle&\ge& \langle S_{\cal A}\rangle\langle
  S_{\cal B}\rangle,
  \label{eq:gks-ineq-two}  
\end{eqnarray}
valid in the ferromagnetic case, $\mathbf{e}=\mathbf{0}$, for any 
${\cal A}, {\cal B}\subseteq{\cal V}$.  

The second GKS inequality (\ref{eq:gks-ineq-two}) can also be
written\cite{Griffiths-1967,Kelly-Sherman-1968} in terms of the
derivative of $\langle S_{\cal A}\rangle$ with respect to
$K_\mathcal{B}$, the dimensional coupling constant corresponding to
the product of spins $S_\mathcal{B}$,
\begin{equation}
  \label{eq:gks-ineq-one-prime}
{  d\langle S_\mathcal{A}\rangle\over d K_\mathcal{B}}\ge 0.
\end{equation}
This implies the monotonicity of any average with respect to all
coupling constants and, as a consequence, the existence of two
extremal Gibbs states describing (generally different)
thermodynamical limit(s) for the Ising model on an infinite hypergraph
$\mathcal{H}=(\mathcal{V},\mathcal{B})$, with free and wired boundary
conditions, respectively.  Namely, one considers an increasing
sequence $\mathcal{V}_t$, $t\in\mathbb{N}$, of sets of vertices,
$\mathcal{V}_t\subsetneq \mathcal{V}_{t+1}\subset
\mathcal{\mathcal{V}}$ which converges weakly to
$\mathcal{V}=\cup_{t\in\mathbb{N}}\mathcal{V}_t$, and the sequence of
sub-hypergraphs $\mathcal{H}_t=(\mathcal{V}_t,\mathcal{B}_t)$ induced
by the sets $\mathcal{V}_t$.  For each $\mathcal{H}_t$, consider also
the hypergraph $\mathcal{H}_t'=(\mathcal{V}_t',\mathcal{B}_t')$,
obtained from $\mathcal{H}$ by contracting all vertices outside
$\mathcal{V}_t$ into one.  Denote the vertex-edge incidence matrices
of $\mathcal{H}_t$ and $\mathcal{H}_t'$ as $\Theta_t^\mathrm{f}$ and
$\Theta_t^\mathrm{w}$, respectively. Here ``f'' and ``w'' stand for
``free'' and ``wired'' boundary conditions in the Ising models
(\ref{eq:prob-distribution}) defined with the help of these matrices.
Clearly, $\mathcal{H}_t$ can be obtained from $\mathcal{H}_{t+1}$ by
reducing some couplings to zero, while $\mathcal{H}_t'$ can be obtained
from $\mathcal{H}_{t+1}'$ by increasing some couplings to infinity.
This implies that for any set of vertices
$\mathcal{A}\subset \mathcal{V}$, and $t$ large enough so that
$\mathcal{A}\subset\mathcal{V}_t$, the averages
$\langle S_\mathcal{A}\rangle_t^\mathrm{f}\le\langle
S_\mathcal{A}\rangle_t^\mathrm{w}$ are, respectively, non-decreasing
and non-increasing with $t$.  They are also bounded, which proves the
existence of the corresponding pointwise limits,
$\langle S_\mathcal{A}\rangle^\mathrm{f}\le\langle
S_\mathcal{A}\rangle^\mathrm{w}$ at any $K$ and $h$.

The two limits are known to coincide \cite{Griffiths-results-1972} for
degree-limited graphs embeddable in $D$-dimensional space, e.g., the
hypercubic lattice $\mathbb{Z}^D$.  Indeed, the increasing sequence of
subgraphs $\mathcal{G}_t=(\mathcal{V}_t,\mathcal{E}_t)$ can be chosen
so that the boundary grows sublinearly with the total number of spins
$|\mathcal{V}_t|$.  Such a property is violated in the case of a
non-amenable graph $\mathcal{G}$, which has a non-zero edge expansion
(Cheeger) constant, $ \iota_E(\mathcal{G})>0$, defined as
\begin{equation}
  \label{eq:cheeger-constant}
  \iota_E(\mathcal{G})\equiv \sup_{\mathcal{W}\subset \mathcal{V}:
    |\mathcal{W}|<\infty}\frac{|\partial_E\mathcal{W}|}{|\mathcal{W}|}, 
\end{equation}
where $\partial_E(\mathcal{W})$ is the set of edges connecting
$\mathcal{W}$ with its complement, $\mathcal{V}\setminus\mathcal{W}$.
The dependence of the critical temperatures (as seen by the
magnetization) on the boundary conditions,
$T_c^\mathrm{w}> T_c^\mathrm{f}$, where the superscripts stand for
``wired'' and ``free'' boundary conditions, respectively, is called
the ``multiplicity'' of critical
points\cite{Wu-hyperb-2000,Schonmann-2001,Haggstrom-Jonasson-Lyons-2002}.
Examples are the infinite $d$-regular trees $\mathcal{T}_d$ (in this
case $T_c^\mathrm{f}=0$, $T_c^\mathrm{w}=(d-1)^{-1}$, see, e.g.,
Ref.~\onlinecite{Lyons-1989}), and the regular $\{f,d\}$ tilings
${\bf H}(f,d)$ of the infinite hyperbolic plane, $df/(f+d)>2$,
where in each vertex $d$ regular $f$-gons meet.  In the latter case
multiplicity of the critical points have been demonstrated for
self-dual graphs, $d=f$, and for graphs with large enough curvature
\cite{Wu-hyperb-2000,Schonmann-2001,Haggstrom-Jonasson-Lyons-2002}.
In Sec.~\ref{sec:hyperbolic} we prove the multiplicity of critical
points for all hyperbolic tilings ${\bf H}(f,d)$ with
$df/(d+f)>2$.

Another important result for the Ising model
(\ref{eq:prob-distribution}) is the duality
transformation\cite{Kramers-Wannier-1941,Wegner-1971}.  In particular,
in the absence of bond disorder, $\mathbf{e}=0$, and at $h=0$, one has
\begin{equation}
  Z_{\mathbf{0}}(\Theta;K)  =
  Z_{\mathbf{0}}(\Theta^*;K^*)\,  2^{r-n_g^*}
\left(\sinh K \cosh K\right)^{n/2} ,
  \label{eq:duality}
\end{equation}
where $K^*$ is the Kramers-Wannier dual of $K$, namely
$\tanh {K}^*=e^{-2 {K}}$, the degeneracy
$n_g^*=r^*-\rank\Theta^*$ ($2^{n_g^*}$ is the number of distinct
ground-state spin configurations in the dual representation), and
$\Theta^*$ is a binary $r^*\times n$ matrix exactly dual to
$\Theta$ (binary rank is used),
\begin{equation}
  \label{eq:matrix-duality}
  \Theta^*\Theta^T=0,\quad \rank \Theta+\rank \Theta^*=n . 
\end{equation}

Notice that in Eq.~(\ref{eq:duality}), and
elsewhere in this work, we simplify the notations by suppressing the
argument corresponding to a zero magnetic field, $h=0$.

Exact duality also works in the presence of sign bond disorder, except
the 
corresponding bonds (``electric charges'') are mapped by duality to
extra factors in front of the exponent, ``magnetic charges''.  The
resulting expression is not positive-definite and thus cannot be
interpreted as a probability measure; instead it is proportional to
the average of a product of the corresponding bonds.  The duality in
this case reads\cite{Wegner-1971}
\begin{equation} {Z_\mathbf{e}(\Theta;K)\over Z_\mathbf{0}(\Theta;K)}=
  \left\langle   \prod_{b\in\mathcal{B}}R_b^{e_b}
  \right\rangle_{\Theta^*;K^*}, \label{eq:em-duality}
\end{equation}
where the average on the right is computed in the dual model with all
bonds ferromagnetic, cf.\ Eq.\ (\ref{eq:bond-product}).

There is a natural notion of equivalence between defects $\mathbf{e}$
that produce identical averages in Eq.~(\ref{eq:em-duality}).  For the
electric charges in the l.h.s., equivalent are any two defects which
differ by a linear combination of rows of $\Theta$,
$\mathbf{e}\simeq\mathbf{e}'=\mathbf{e}+\boldsymbol\alpha \Theta$,
where $\boldsymbol\alpha$ is a length-$r$ binary vector.  Such defects
are related by Nishimori's spin-glass \emph{gauge
  symmetry}\cite{Nishimori-book} generated by local spin flips
$\alpha_v\in\Ftwo$, $v\in\mathcal{V}$, and simultaneous update of the
components of $\mathbf{e}$ on the adjacent bonds,%
\begin{equation}
  \label{eq:gauge-transformation}
  S_v\to (-1)^{\alpha_v}S_v,\quad e_b\to e_b'\equiv
  e_b+\sum\nolimits_v\alpha_v\Theta_{vb}. 
\end{equation}
For such a defect $\mathbf{e}$, it is convenient to introduce an
invariant distance $d_\mathbf{e}$, the minimum number of flipped bonds
among all defects in the same equivalence class,
\begin{equation}
  d_\mathbf{e}\equiv d_\mathbf{e}(\Theta)=\min_{\boldsymbol\alpha}
  \wgt(\mathbf{e}+{\boldsymbol\alpha}\Theta),\label{eq:defect-distance} 
\end{equation}
where $\wgt(\mathbf{e})$ is the Hamming weight.
An identical equivalence relation for the magnetic charges which
define the product of spins in the r.h.s.\ of
Eq.~(\ref{eq:em-duality}) can be interpreted as a result of
introducing a product of (dual) bonds that form a cycle, i.e., does not
change the spins that actually enter the average.

For a finite system and a finite $K>0$, both sides of
Eq.~(\ref{eq:em-duality}) are strictly positive.  The logarithm of the
l.h.s.\ is proportional to the free energy increment due to the
addition of the defect,
\begin{equation}
\delta_\mathbf{e}\equiv \delta_\mathbf{e}(\Theta;K)\equiv \ln
Z_\mathbf{0}(\Theta;K)-\ln Z_\mathbf{e}(\Theta;K);\label{eq:defect-delta}
\end{equation}
in turn, it is proportional to dimensionless 
 defect tension%
\begin{equation}
  \label{eq:defect-tension}
  \tau_\mathbf{e}\equiv  \tau_\mathbf{e}(\Theta;K)
  \equiv \delta_\mathbf{e}(\Theta;K)/d_\mathbf{e}.
\end{equation}
Respectively, the scaling of the spin average in the r.h.s.\ of
Eq.~(\ref{eq:em-duality}) with the minimum number of bonds in the
product is called the \emph{area-law} exponent,
\begin{equation}
  \alpha_\mathbf{e}\equiv \alpha_\mathbf{e}(\Theta^*;K^*)
  =-d_\mathbf{e}^{-1}\ln \left\langle
    \prod_{b\in\mathcal{B}}R_b^{e_b}\right\rangle_{\Theta^*;K^*}.
  \label{eq:area-law} 
\end{equation}
Second GKS inequality (\ref{eq:gks-ineq-two}) implies subadditivity, 
\begin{equation}
  \label{eq:alpha-subadditivity}
    d_{\mathbf{e}_1+\mathbf{e}_2}\alpha_{\mathbf{e}_1+\mathbf{e}_2}\le
    {d_{\mathbf{e}_1}\alpha_{\mathbf{e}_1}+d_{\mathbf{e}_2}\alpha_{\mathbf{e}_2}}.
\end{equation}
Thus electric-magnetic duality (\ref{eq:em-duality}) also implies an
exact relation between the defect tension and area-law exponent in a
pair of dual models,
\begin{equation}
  \label{eq:em-duality-tension}
  \tau_\mathbf{e}(\Theta;K)=\alpha_\mathbf{e}(\Theta^*;K^*).
\end{equation}
Combined with Eq.~(\ref{eq:alpha-subadditivity}), this implies
subadditivity for defect free energy cost
\begin{equation}
  \label{eq:delta-subadittivity}
    \delta_{\mathbf{e}_1+\mathbf{e}_2}\le
    \delta_{\mathbf{e}_1}+\delta_{\mathbf{e}_2}.
\end{equation}
In the special case of a model with two-body couplings defined on a
graph ${\cal G}=({\cal V},{\cal E})$, a single correlation decay
exponent can be defined in terms of pair correlations,
\begin{equation}
\alpha\equiv\alpha(\mathcal{G};K)=
\inf_{i,j\in\mathcal{V}}\left[-{\ln\langle
S_iS_j\rangle\over d_{ij}}\right],\label{eq:asymptotic-alpha} 
\end{equation}
where $d_{ij}$ is the graph distance between  $i$ and $j$.
Subadditivity (\ref{eq:alpha-subadditivity}) implies that the value of
$\alpha$ corresponds to that for pairs with large $d_{ij}$.

We are interested in the Ising models (\ref{eq:prob-distribution})
with few-body couplings.
More specifically, we consider \emph{weight-limited} Ising models
with vertex and bond degrees bounded by some fixed $\ell$ and $m$,
respectively, so that $d_v\le m$, $v\in{\cal V}$, and $d_b\le \ell$,
$b\in{\cal B}$.  With fixed $\ell$ and $m$, we call such a model
$(\ell,m)$-sparse.  This refers to the sparsity of the corresponding
coupling matrix $\Theta$: $\ell$ and $m$, respectively, are the
maximum weights of a column and of a row of $\Theta$.

Further, we would like to consider models whose duals are in the same
class of weight-limited Ising models, with some maximum vertex,
$\ell^*$, and bond, $m^*$, degrees.  However, such a condition would
be very restrictive if one insists on the exact
duality~(\ref{eq:matrix-duality}).  For example, in the case of the
square-lattice Ising model with periodic boundary conditions on an
$L\times L$ square ($\ell=2$ and $m=4$), the dual model can be chosen
to have the same vertex and bond degrees, $\ell^*=2$ and $m^*=4$,
except for $k=2$ additional summations over periodic/antiperiodic
boundary conditions in each direction.  These summations can be
introduced as additional spins entering $d_v\ge L$ bonds, where the
lower bound is the length of the shortest domain wall on the
$L\times L$ square-lattice tiling of a torus.  The two
additional summations give no contribution to the asymptotic free
energy density at $L\to\infty$, both in the low- and high-temperature
phases, and are often ignored.

Such a \emph{weak} duality with additional defects for models on
locally planar graphs can be generalized by considering a pair of
weight-limited binary matrices with $n$ columns each, $G$ and $H$,
such that their rows be mutually orthogonal, $G\,H^T=0$.  Since we do
not require exact duality (\ref{eq:matrix-duality}), there are exactly
\begin{equation}
k\equiv n-\rank G-\rank H
\label{eq:CSS-k}
\end{equation}
distinct \emph{defect} vectors $\mathbf{c}_i\in\mathbb{F}_2^n$,
$i\in\{1,\ldots,k\}$, which are orthogonal to rows of $H$ and whose
non-trivial linear combinations are linearly-independent from rows of
$G$.

Just as for the spin glasses on locally planar graphs, the matrix $H$
can be used to define \emph{frustration},
$\mathbf{s}\equiv \mathbf{e} \,H^T$, a gauge-invariant
characteristic of bond disorder.
As common in spin-glass theory\cite{Nishimori-book}, we will consider
independent identically-distributed (i.i.d.)\ components of the
quenched disorder vector $\mathbf{e}$, such that $e_b=1$ with
probability $p$.  The corresponding averages are denoted with
square brackets, $[\,\boldsymbol\cdot\,]_p$.

In theory of quantum error correcting
codes\cite{gottesman-thesis,Calderbank-1997,Nielsen-book,preskill-book},
a pair of binary matrices with orthogonal rows, $G \,H^T=0$, can be
used to define a
Calderbank-Shor-Steane\cite{Calderbank-Shor-1996,Steane-1996} (CSS)
stabilizer code ${\cal Q}(G,H)$ which encodes $k$ qubits in $n$, see
Eq.~(\ref{eq:CSS-k}).  Such a quantum code has a convenient
representation in terms of \emph{classical} binary codes.  Given a
matrix $G$ with $n$ columns, one defines the classical code
$\mathcal{C}_G\subseteq\Ftwo^n$, a linear space of dimension $\rank G$
generated by the rows of $G$.  One also defines the corresponding
\emph{dual} code $\mathcal{C}_G^\perp$ of all vectors in
$\mathbb{F}_2^n$ orthogonal to rows of $G$; such a code is generated
by the corresponding dual matrix (\ref{eq:matrix-duality}),
$\mathcal{C}_G^\perp\equiv\mathcal{C}_{G^*}$.  By orthogonality, we
necessarily have $\mathcal{C}_H\subseteq \mathcal{C}_G^\perp$ and
$\mathcal{C}_G\subseteq \mathcal{C}_H^\perp$, where equality is
achieved when the two matrices are exact dual of each other, in which
case $k=0$.  The defect vectors $\mathbf{c}$ are non-zero CSS
\emph{codewords} of $G$ type,
$\mathbf{c}\in\mathcal{C}_H^\perp\setminus \mathcal{C}_G$; there are
exactly $2^k-1$ inequivalent (mutually
\emph{non-degenerate}\cite{Calderbank-1997}) vectors of this type.
Similarly, there are also $2^k-1$ inequivalent $H$-type vectors
$\mathbf{b}$ in $\mathcal{C}_G^\perp\setminus \mathcal{C}_H$, where
equivalence is defined in terms of rows of $H$,
$\mathbf{b}'\simeq \mathbf{b}$ iff
$\mathbf{b}'=\mathbf{b}+\alpha^T H$.  For any quantum code, important
parameters are its rate, $R\equiv k/n$, and the distance,
$d\equiv \min (d_G, d_H)$,
\begin{equation}
d_G\equiv \min_{\mathbf{c}\in \mathcal{C}_H^\perp\setminus
  \mathcal{C}_G}
\wgt(\mathbf{c}),\quad
d_H\equiv \min_{\mathbf{b}\in \mathcal{C}_G^\perp\setminus
  \mathcal{C}_H}
\wgt(\mathbf{b}).
\label{eq:CSS-dist}
\end{equation}

As yet another interpretation of the algebraic structure in the pair
of weakly-dual Ising models with vertex-bond incidence matrices $G$
and $H$ of dimensions $r\times n$ and $r'\times n$, respectively,
consider a two-chain complex $\Sigma\equiv \Sigma(G,H)$,%
\begin{equation}
\Sigma:\;C_2\equiv\Ftwo^{r'}\stackrel{\partial_2}\rightarrow
C_1\equiv \Ftwo^n\stackrel{\partial_1}\rightarrow
C_0\equiv\Ftwo^r,\label{eq:chain-complex}  
\end{equation}
where the modules $C_j$, $j\in\{0,1,2\}$ are the linear spaces of
binary vectors with dimensions $r$, $n$, and $r'$, respectively, and
the boundary operators $\partial_1$ and $\partial_2$ are two linear
maps defined by the matrices $G$ and $H^T$.  The required composition
property, $\partial_1\circ \partial_2=0$, is guaranteed by the
orthogonality between the rows of $G$ and $H$. The number of
independent defect vectors (\ref{eq:CSS-k}) is exactly the rank of the
first homology group $H_1(\Sigma)$.

\section{Results}
\label{sec:results}

\subsection{Properties of specific homological difference}
\label{sec:general}

We first quantify the effect of homological defects on the properties
of general Ising models.  To this end, given a pair of binary matrices
$G$ and $H$ with $n$ columns each and mutually orthogonal rows,
$GH^T=0$, consider the specific homological difference
\cite{Kovalev-Prabhakar-Dumer-Pryadko-2018} (per bond),
\begin{eqnarray}
  \nonumber
  \Delta f_{\mathbf{e}}
  &\equiv&   \Delta f_{\mathbf{e}}(G,H;K)\\
  &  =&{1\over n}\left\{\ln Z_{\mathbf{e}}(H^*;K)-\ln
        Z_{\mathbf{e}}(G;K)\right\},\quad
  \label{eq:homo-difference}
\end{eqnarray}
where, to fix the normalization, the dual matrix $H^*$ [see
Eq.~(\ref{eq:matrix-duality})] is constructed from $G$ by adding
exactly $k$ row vectors\footnote{Notice that any other construction of
  the dual matrix would at most change the partition function
  multiplicatively by a power of two.}, linearly-independent
inequivalent codewords
$\mathbf{c}\in{\cal C}_H^\perp\setminus{\cal C}_G$.  This quantity
satisfies the inequalities\cite{Kovalev-Prabhakar-Dumer-Pryadko-2018}
\begin{equation}
\begin{aligned}
  0\le \Delta f_\mathbf{0}(G,H;K)&\le \Delta f_\mathbf{e}(G,H;K),\\
  \Delta f_\mathbf{0}(G,H;K)&\le R\ln2,
\end{aligned}\label{eq:homo-generic-bounds}
\end{equation}
where $R\equiv k/n$, and $k$ is the homology rank given by
Eq.~(\ref{eq:CSS-k}).  The lower and the upper bounds are saturated,
respectively, in the limits of zero and infinite temperatures.  In
addition, in the absence of disorder, the specific homological
difference is a non-increasing function of $K$ (and non-decreasing
function of $T=J/K$),
\begin{equation}
  \label{eq:monotonicity}
  {d\over dK} \Delta f_\mathbf{0}(G,H;K)\le0.
\end{equation}

Our
starting point is the following Theorem (related to Theorem 2 in
Ref.~\onlinecite{Kovalev-Prabhakar-Dumer-Pryadko-2018}), proved in
Appendix \ref{app:ths-convergence-proof}:
\begin{restatable}{theorem}{ltsconvergence}
  \label{th:lts-convergence}
  Consider a sequence of pairs of weakly dual Ising models defined by
  pairs of finite binary matrices with mutually orthogonal rows,
  $G_tH_t^T=0$, $t\in\mathbb{N}$, where row weights of each $H_t$ do
  not exceed a fixed $m$.  In addition, assume that the sequence of
  the CSS distances $d_{G_t}$ is increasing.  Then the sequence
  $\Delta f_t\equiv [\Delta f_\mathbf{e}(G_t,H_t;K)]_p$,
  $t\in\mathbb{N}$, converges to zero in the region
  \begin{equation}
(m-1)[e^{-2K}(1-p)+e^{2K}p]<1.\label{eq:low-T-low-p-region}
\end{equation}
\end{restatable}
\noindent\textsc{Remarks:} \ref{th:lts-convergence}-1.\ The bound in Theorem
\ref{th:lts-convergence} guarantees the existence of a
\emph{homological} region where $\Delta f_t$ converges to zero.
Generally, such a region may be wider than what is granted by the
sufficient condition (\ref{eq:low-T-low-p-region}).  We will denote
$K_h(G,H;p)$ the smallest $K>0$ such that the series $\Delta f_t$
converges to zero at any $K'>K$.  The corresponding temperature,
$T_h(G,H;p)\equiv J/K_h(G,H;p)$, is the upper boundary for the
homological region at this $p$.  Eq.~(\ref{eq:low-T-low-p-region})
implies, in particular, that $K_h(G,P;0)\ge \ln(m-1)/2$.  

\smallskip\noindent\ref{th:lts-convergence}-2.\ In the homological
region, the sequence of the average free energy densities
$[f_\mathbf{e}(G_t,K)]_p$ converges iff the sequence
$[f_\mathbf{e}(H_t^*,K)]_p$ converges, and the corresponding limits
coincide.  

\smallskip\noindent\ref{th:lts-convergence}-3.\ In analogy with
Eq.~(\ref{eq:defect-tension}), we introduce the defect tension in the
presence of disorder,
\begin{equation}
\tau_{\mathbf{c},\mathbf{e}}\equiv
\tau_{\mathbf{c},\mathbf{e}}(G;K)\equiv 
 {1\over d_\mathbf{c}} \left\{F_{\mathbf{e}+\mathbf{c}}(G;K)
   -F_{\mathbf{e}}(G;K)\right\}, 
  \label{eq:disorder-averaged-tension} 
\end{equation}
where $d_\mathbf{c}\ge d_G$ is the minimum weight of the codeword
equivalent to
$\mathbf{c}\in\mathcal{C}_H^\perp\setminus \mathcal{C}_G$.  While the
tension (\ref{eq:disorder-averaged-tension}) is not necessarily
positive, it satisfies the inequalities
\begin{equation}
  \label{eq:tension-generic-bounds}
| \tau_\mathbf{c,e}|\le\tau_\mathbf{c,0}\le 2K.
\end{equation}
We also define the weighted average defect tension,
\begin{equation}
  \label{eq:weighted-average-defect-tension}
  \bar\tau_p\equiv  {\sum_{\mathbf{c}\not\simeq\mathbf{0}}
   d_\mathbf{c}[\tau_{\mathbf{c},\mathbf{e}}]_p\over
   \sum_{\mathbf{c}\not\simeq\mathbf{0}}d_\mathbf{c}}, 
\end{equation}
where the average is over disorder and the $2^k-1$ non-trivial defect
classes.  This quantity satisfies the following bound in terms of the
average homological difference,
\begin{equation}
  \label{eq:average-tension}
  \zeta\, \bar\tau_p\ge R\ln2-[\Delta f_{\bf e}]_p,
\end{equation}
where the dimensionless constant $\zeta\le1/2$, see
Eq.~(\ref{eq:zeta-constant}) in the Appendix.  In the homological
phase this gives $\bar\tau_p\ge 2R\ln2$.  (A related bound was
previously obtained for the boundary of \emph{decodable} phase in
Ref.~\onlinecite{Kovalev-Pryadko-SG-2015}.)

\smallskip

In the absence of disorder, $\mathbf{e}=\mathbf{0}$, the specific
homological difference is
self-dual\cite{Kovalev-Prabhakar-Dumer-Pryadko-2018}, up to an
exchange of the matrices $G$ and $H$, and an additive constant,
\begin{equation}
  \label{eq:homological-diff-duality}
\Delta f_\mathbf{0}
(G,H;K)=R\ln 2-\Delta f_\mathbf{0}
(H,G;K^*).  
\end{equation}
Comparing with the general inequalities~(\ref{eq:homo-generic-bounds}), one
sees that a point close to the lower bound is mapped to a point close
to the corresponding upper bound.  This implies a version of Theorem
\ref{th:lts-convergence} applicable for high temperatures:
\begin{restatable}{theorem}{htsconvergencehomo}
  \label{th:hts-convergence-homo}
  Consider a sequence of pairs of weakly dual Ising models defined by
  pairs of finite binary matrices with mutually orthogonal rows,
  $G_tH_t^T=0$, $t\in\mathbb{N}$, where row weights of each $G_t$ do
  not exceed a fixed $m$, CSS distances $d_{H_t}$ are increasing with
  $t$, and the sequence of CSS rates $R_t\equiv k_t/n_t$ converges,
  $\lim_t R_t=R$.  Then, for any $K\ge0$ such that $(m-1)\tanh K<1$,
  the sequence $\Delta f_t\equiv [\Delta f_\mathbf{e}(G_t,H_t;K)]_p$,
  $t\in\mathbb{N}$, converges to $R\ln 2$.
\end{restatable}
\noindent Remarks: \ref{th:hts-convergence-homo}-1.\ Since duality is
used in the proof, we had to switch the conditions on the matrices
$G_t$ and $H_t$.  Similarly, the bound for $\tanh K$ is the
Kramers-Wannier dual of that in Eq.~(\ref{eq:low-T-low-p-region}) at
$p=0$.

\smallskip\noindent\ref{th:hts-convergence-homo}-2.\ We will call the
temperature region where the sequence $\Delta f_t$ in Theorem
\ref{th:hts-convergence-homo} converges to $R\ln2$  the \emph{dual
  homological} region.  Given that the homological region in the
absence of disorder extends throughout the interval $K\ge K_h(G,H)$,
the corresponding interval for the dual homological region is
$K\le K_h^*(H,G)$, where $K^*$ denotes the Kramers-Wannier dual, see
Eq.~(\ref{eq:duality}).  Respectively, $T_h^*(H,G)\equiv J/K_h^*(H,G)$
is the low temperature boundary of the dual homological region at
$p=0$.

\smallskip\noindent\ref{th:hts-convergence-homo}-3.\ In the dual
homological region, the sequence of the free energy densities
$f_\mathbf{0}(H_t^*,K)$ converges iff the sequence
$f_\mathbf{0}(G_t,K)$ converges, and the corresponding limits
$f_{H^*}(K)$ and $f_{G}(K)$ satisfy
  \begin{equation}
    f_{G}(K)=f_{H^*}(K)+R\ln2.\label{eq:hts-homo}
  \end{equation}
\smallskip

Notice that when both sets of matrices $H_t$ and $G_t$,
$t\in\mathbb{N}$, have bounded row weights, the same sequence
$\Delta f_\mathbf{0}(G_t,H_t;K)$ converges to zero in the homological
region, $K\ge K_h(G,H)$, and to $R\ln2$ in the dual homological
region, $K\le K_h^*(H,G)$.  Since the magnitude of the derivative of
the free energy density with respect to $K$ (proportional to the
energy per bond) is bounded, for any $R>0$ this implies the existence
of a minimum gap between the boundaries of the homological and the
dual homological regions.  We have the inequality
\begin{equation}
  \label{eq:ineq-three}
   K_h(G,H)-K_h^*(H,G)\ge R\ln2.
\end{equation}

\subsection{Free energy analyticity and convergence}
\label{sec:analyticity}

The end points, $T_h(G,H)$ and $T_h^*(H,G)$ of the two flat
regions in the temperature dependence of the homological difference
$\Delta f_\mathbf{0}$ are clearly the points of singularity.  What is
the relation between these points and the singular points of the
limiting free energy density in individual models, which are usually
associated with phase transitions? 

To establish such a relation, let us analyze the convergence of free
energy density and the analyticity of the corresponding limit as a
function of parameters.  To this end, consider the high-temperature
series (HTS) expansion of the free energy density (\ref{eq:Z}),
\begin{equation}
  \label{eq:hts-f-general}
  f_\mathbf{e}(\Theta;K,h)\equiv \sum_{s=1}^\infty
  \kappa_\mathbf{e}^{(s)}(\Theta;J,h')\,{\beta^s\over s!},    
\end{equation}
where both parameters are scaled with the inverse temperature,
$K\equiv \beta J$ and $h\equiv \beta h'$.  The coefficient in front of
$\beta^s$ is proportional to an order-$s$ cumulant of energy; it is a
homogeneous polynomial of the variables $h'$ and $J$ of degree $s$.  A
general bound on high-order cumulants from
Ref.~\onlinecite{Feray-Meliot-Nikeghbali-2013} gives the following
\begin{restatable}{statement}{htsconvergence}
  \label{th:hts-convergence} Consider any model in the form
  (\ref{eq:prob-distribution}), with an $(\ell,m)$-sparse $r\times n$
  coupling matrix $\Theta$.  The coefficients of the HTS expansion of
  the free energy density satisfy
  \begin{equation}
    \label{eq:cumulant-bound-specific}
    |\kappa_\mathbf{e}^{(s)}(\Theta;J,h')|\le 2^{s-1}s^{s-2} \,C\,
    (\Delta+1)^{s-1}A^s, 
  \end{equation}
  where $A\equiv \max(|J|,|h'|)$ and (\textbf{a}) with $J$ and $h'$
  both non-zero, $\Delta=\ell m$ and $C=r/n+1$, while (\textbf{b}) with $h'=0$,
  $\Delta=(\ell-1)m$ and $C=1$. 
\end{restatable}
Such a bound implies the absolute convergence of the HTS in a finite
circle in the complex plane of $\beta$ and, thus, the 
analyticity of
$f_\mathbf{e}(\Theta;K,h)$ and all of its derivatives as a function of
both variables in a finite region with $|K|$ and $|h|$ small enough,
in any finite $(\ell,m)$-sparse Ising model, 
at any given configuration of flipped bonds $\mathbf{e}$.  The same is
true for the average free energy $[f_\mathbf{e}(\Theta;J,h')]_p$.

In this region, at $p=0$, convergence and analyticity of the limiting
free energy density for models defined by a sequence of binary
matrices $\Theta_t$, $t\in \mathbb{N}$, is equivalent to existence of
the (pointwise) limit $\lim_t\kappa_\mathbf{0}^{(s)}(G_t;J,h')$ for
the individual coefficients (remember, each of them is a homogeneous
two-variate polynomial of degree $s$).  With the help of the cluster
theorem for the HTS coefficients, the existence of the limit can be
guaranteed by the Benjamini-Schramm
convergence\cite{Benjamini-Schramm-2001} of the corresponding Tanner
graphs, see
Refs.~\onlinecite{Borgs-Chayes-Kahn-Lovasz-2013,Lovasz-2016} for the
corresponding discussion for general models with up to two-body
couplings.  For our present purposes, the following subsequence
construction at $h=0$ is sufficient:%
\begin{restatable}{corollary}{htsanalyticseq}
  \label{th:hts-analytic-seq} Any infinite sequence of
  ($\ell,m$)-sparse Ising models, specified in terms of the matrices
  ${\Theta}_j$, $j\in \mathbb{N}$, has an infinite subsequence
  ${\Theta}_{j(t)}$, $t\in \mathbb{N}$, where
  $j:\mathbb{N}\to \mathbb{N}$ is strictly increasing, such that
  (\textbf{a}) for each $s$, the sequence of the coefficients
  $\kappa_\mathbf{0}^{(s)}(\Theta_t;J,0)$ converges with $t$, and
  (\textbf{b}) the sequence of free energy densities
  $f(\Theta_{j(t)};K)$ has a limit, $\varphi_\Theta(K)$, which is an
  analytic function of $K$ in the interior of the circle
  $|K|\le \{2e\,[(\ell-1) m+1]\}^{-1}$. 
  Here $e$ is the base of natural logarithm.
\end{restatable}

\noindent Remarks: \ref{th:hts-analytic-seq}-1.\ Similar analyticity
bounds apply to a very general class of $(\ell,m)$-sparse models with
up to $\ell$-body interactions, where each variable is included in up
to $m$ interaction terms, and magnitudes of different interaction
terms are uniformly bounded: the dependency graph used in the proof
can be used in application to all such models.  Examples include a
variety of discrete models, e.g., Potts and clock models with few-body
couplings, as well as compact continuous models with various symmetry
groups, Abelian and non-Abelian, where interaction terms are
constructed as traces of products of unitary matrices.  This is a
generalization of the ``right'' convergence established for models
with two-body couplings ($\ell=2$) in
Refs.~\onlinecite{Borgs-Chayes-Kahn-Lovasz-2013,%
  Lovasz-2016}.

\smallskip\noindent\ref{th:hts-analytic-seq}-2.\ The 
subsequence construction is not necessary in the special case where
the Tanner graphs defined by the bipartite matrices $\Theta_t$ are
transitive, with weak infinite-graph limit $\Theta$
and a center $0\in\mathcal{V}(\Theta)$, such that a ball of radius
$\rho_t$ in $\Theta_t$ is isomorphic to the ball of the same radius
centered around $0$ in $\Theta$; here the sequence of the radii is
increasing, $\rho_{t+1}>\rho_t$, $t\in\mathbb{N}$.  In this case the cluster
theorem\cite{Domb-Green-book} guarantees that the coefficients
$\kappa_s(\Theta_t)$ do not depend on $t$ for $\rho_t>s$.
\smallskip

To make precise statements applicable outside of the convergence
radius of the high-temperature series, we need to ensure that a
sequence of free energy densities converges.  The question of
convergence for a general sequence of Ising models being far outside the
scope of this work, we will assume the use of yet another subsequence
construction to guarantee the existence of the thermodynamical limit
for the free energy density.  This is based on the following Lemma
proved in Appendix \ref{app:subsequencelemma}.
\begin{restatable}{lemma}{subsequencelemma}
  \label{th:subsequence-lemma}
  Consider a sequence of $r_t\times n_t$ binary matrices $\Theta_t$,
  where $0<r_t\le n_t$, and $t\in\mathbb{N}$.  For any $M>0$, define a
  closed interval $I_M\equiv [0,M]$.  (a) There exists a subsequence
  $ \Theta_{t(i)}$, $i\in\mathbb{N}$, where the
  function $t:\mathbb{N}\to\mathbb{N}$ is strictly increasing,
  $t(i+1)>t(i)$ for all $i\in\mathbb{N}$, such that the sequence of
  Ising free energy densities converges for any $K\in I_M$,
  ${f}_i(K)\equiv f_\mathbf{0}(\Theta_{t(i)};K)\to f(K)$.  (b)
  The limit $f(K)$ is a continuous non-increasing concave
  function with left and right derivatives uniformly bounded,
  \begin{equation}
    -1\le f_+'(K)\le f_-'(K)\le0,\label{eq:left-right-bnd}  
  \end{equation}
for all $K\in I_M$.
\end{restatable}

Let us now assume that we have a sequence of pairs of weakly-dual
weight-limited Ising models which (a) satisfy the conditions of
Theorems \ref{th:lts-convergence} and \ref{th:hts-convergence-homo}
with the asymptotic rate $R$, (b) such that the coefficients of the
corresponding HTSs converge, so that the sequences of free energy
densities $f(G_t;K)$ and $f(H_t;K)$ both converge to analytic
functions, $\varphi_G(K)$ and $\varphi_H(K)$ respectively, at $|K|$
sufficiently small (Corollary \ref{th:hts-analytic-seq}), and, in
addition, (c) the sequences of free energy densities both converge on
an interval of real axis $I_M$, with $M>\ln(m-1)/2$.

The interval in (c) is such that Theorems \ref{th:lts-convergence} and
\ref{th:hts-convergence-homo} can be used to extend the convergence to
the entire real axis; we denote the corresponding limits $f_G(K)$ and
$f_H(K)$.  The continuity of the functions $f_G(K)$ and $f_H(K)$ (and
the corresponding duals), along with the inequality
(\ref{eq:monotonicity}) which also survives the limit, guarantee that
in the range of temperatures between the homological and the dual
homological regions, $ T_h(G,H)<T<T_h^*(H,G)$, the specific
homological difference $\Delta f(K)\equiv f_G(K)-f_{H^*}(K)$ satisfies
the strict inequality
\begin{equation}
  \label{eq:homo-difference-intermediate}
  0<\Delta f(K)<R\ln2.
\end{equation}

Notice that the existence of the limit on the real axis does not
guarantee analyticity which is only guaranteed by condition (b) in a
finite vicinity of $K=0$.  Hereafter, we will assume that $f_G(K)$ is
analytic on the interval $0\le K<K_c(G)$.  That is, for any
$\epsilon>0$, there exists a simply-connected open complex region
$\Omega_\epsilon\in\mathbb{C}$ which includes the union of the circle
of convergence of HTS for $\varphi_G(K)$ from Corollary
\ref{th:hts-analytic-seq} and the interval $I_{M}$,
$M=K_c(G)-\epsilon$, such that the sum of HTS series $\varphi_G(K)$
can be analytically continued to $\Omega_\epsilon$, and the result
coincides with the limit $f_G(K)$ on the real axis, $K\in I_M$.
Further, we will assume that $K_c(G)$ is the largest value at which
this is possible.  Such a threshold may arise either (i) because
$K_c(G)$ is a singular point of $\varphi_G(K)$, e.g., the intersection
of the natural boundary of $\varphi_G(K)$ with the real axis, or (ii)
the limit on the real axis, $f_G(K)$, starts to deviate from the
result of the analytic continuation.  In either case, this guarantees
that the limit on the real axis, $f_G(K)$, has a singular point of
some sort at $K_c(G)$.

According to this definition, $T_c(G)=J/K_c(G)$ is the
highest-temperature point of non-analyticity of the limiting free
energy density $f_G(K)$; $f_G(J/T)$ is analytic for $T>T_c(G)$.  By
duality and Theorem \ref{th:lts-convergence}, $f_G(K)$ is also
analytic at low temperatures.  We denote $T'_c(G)\le T_c(G)$ the
lowest-temperature singular point of $f_G(J/T)$.  

We make similar assumptions about the properties of the limiting free
energy density $f_H(K)$, and use similar definitions of the critical
temperatures $T_c'(H)\le T_c(H)$ for $f_H(J/T)$.  We will also use the
dual functions, $f_{G^*}(K)$ and $f_{H^*}(K)$, which coincide with
$f_G(K^*)$ and $f_H(K^*)$ up to an addition of analytic functions of
$K$, see Eq.~(\ref{eq:duality}).  The corresponding lowest- and
highest-temperature singular points are exchanged by duality, e.g.,
$T_c'(H^*)=T_c^*(H)$, $T'_c(H)=T_c^*(H^*)$.  Convergence of
$\Delta f(G_t,H_t;K)$ to zero implies that $f_G(K)=f_{H^*}(K)$ for
$K>K_h(G,H)$, thus $f_G(K)$ is an analytic function in a complex
vicinity of any $K>\max\biglb(K_c^*(H),K_h(G,H)\bigrb)$.
Equivalently,
\begin{equation}
  T_c'(G)\ge \min\biglb(T_c^*(H)=T_c'(H^*),T_h(G,H)\bigrb).
  \label{eq:lower-tc-ineq}
\end{equation}

Once we are assured of convergence of the homological difference, the
first observation is that the limit, $\Delta f(K)$, is necessarily a
strictly convex function at $T_h(G,H)$, and a strictly concave
function at $T_h^*(H,G)$, the singular points which are also the
boundaries of the region separating the dual homological region at
small $K$ and the homological region at large $K$.  On the other hand,
both $f_{G}(K)$ and $f_{H^*}(K)$ are concave functions.  Therefore,
the convexity at $T_h(G,H)$ must originate from $f_{H^*}(G,H)$.

Unfortunately, this does not guarantee that $T_h(G,H)$ be a singular
point of $f_{H^*}(K)$.  A higher-order phase transition, with a
continuous specific heat but discontinuity or divergence in its first
or higher derivative, cannot be eliminated on the basis of the general
thermodynamical considerations alone.  Therefore, we formulate Theorem
\ref{th:temperature-inequalities} below (proved in Appendix
\ref{app:temperature-inequalities}) with a list of
independently-sufficient conditions.

\begin{restatable}{theorem}{temperatureinequalities}
  \label{th:temperature-inequalities}
  Let us assume that any one of the following Conditions is true:
  \begin{enumerate}
  \item The transition at $T_c'(G)$ is discontinuous or has a divergent
    specific heat;
  \item The derivative of $\Delta f(K)=f_G(K)-f_{H^*}(K)$ is
    discontinuous at $K_h\equiv K_h(G,H)$, or the derivative of
    $\Delta f(K)$ is continuous at $K_h$, but its second derivative
    diverges at $K_h$;
  \item Summation over  homological defects does not increase the
    critical temperature, $T_c(G^*)\le T_c(H)$.
  \end{enumerate}
  Then the Kramers-Wannier dual of the critical temperatures $T_c(H)$
  satisfies
  \begin{eqnarray}
    \label{eq:ineq-two}
    T_c^*(H)\le T_h(G,H).
  \end{eqnarray}
\end{restatable}

\noindent\textsc{Remarks:} \ref{th:temperature-inequalities}-1.\
We are making the same assumptions about the properties of $f_H(K)$,
which gives $T_c^*(G)\le T_h(H,G)$.  Combining with
Eq.~(\ref{eq:ineq-three}), we have 
\begin{equation}
  \label{eq:multiplicity-ineq}
  K_c(H^*)-K_c(G)\ge R\ln2.
\end{equation}
This implies a strict inequality, $T_c(G)>T_c(H^*)$, when the
homological rank scales extensively, $R>0$, which is superficially
similar to the multiplicity of critical points on nonamenable infinite
graphs\cite{Wu-hyperb-2000,Schonmann-2001,Haggstrom-Jonasson-Lyons-2002},
see Sec.~\ref{sec:background}.  The difference is that our critical
temperatures correspond to points of non-analyticity of the limiting
free energy density in zero magnetic field; we do not have a direct
connection to magnetic transitions.%

\smallskip\noindent\ref{th:temperature-inequalities}-2.\ It is known
that stabilizer codes with generators local in $\mathbb{Z}^D$ and
divergent distances have asymptotically zero
rates\cite{Bravyi-Terhal-2009,Bravyi-Poulin-Terhal-2010}.  This is
perfectly consistent with the known fact that weight-limited models
local in $\mathbb{Z}^D$ have well-defined thermodynamical limits,
independent of the boundary conditions\cite{Griffiths-results-1972}.
For example, inequality (\ref{eq:multiplicity-ineq}) with $R=0$ is
saturated in the case of planar self-dual Ising models, where the
transition is in the self-duality point, which is the only
non-analyticity point of the free energy density. 

\smallskip\noindent\ref{th:temperature-inequalities}-3.\ Most
important application of Theorem \ref{th:temperature-inequalities} and
Eq.~(\ref{eq:multiplicity-ineq}) are few-body Ising models
that correspond to finite-rate quantum LDPC codes with distances
scaling as a power of the code length $n$, $d\ge An^\alpha$ with
$A,\alpha>0$.  Examples are quantum hypergraph-product (QHP) and related
codes\cite{Tillich-Zemor-2009,Kovalev-Pryadko-Hyperbicycle-2013}, and
higher-dimensional hyperbolic codes\cite{Guth-Lubotzky-2014}.  Because
of higher-order couplings, generic mean-field theory gives a
discontinuous transition, which is the case of Condition 1 in Theorem
\ref{th:temperature-inequalities}.  The discontinuous nature of the
transition has been verified numerically for one class of QHP
codes\cite{Kovalev-Prabhakar-Dumer-Pryadko-2018}.

\smallskip\noindent\ref{th:temperature-inequalities}-4.\ Ising models
on expander graphs are known to have mean-field
criticality\cite{Montanari-Mossel-Sly-2012,Schonmann-2001}.  A
combination of an analytic $f_{H^*}(K)$ and a finite specific heat
jump in $f_{G}(K)$ at $K_h(G,H)$ is not eliminated by the Conditions 1
or 2.  We discuss the important case of Ising models on hyperbolic
graphs in the next Section.

\smallskip\noindent\ref{th:temperature-inequalities}-5.\ GKS
inequalities imply that any spin average satisfies
$\langle S_{\cal A}\rangle_{G;K}\ge \langle S_{\cal
  A}\rangle_{H^*;K}$.  Physically, this ought to be sufficient to
guarantee Condition 3, but we are not aware of a general proof. 


\subsection{Application to models on hyperbolic graphs}
\label{sec:hyperbolic}
\subsubsection{Bounds for infinite-graph transition temperatures}

While the inequalities (\ref{eq:ineq-three}) are
(\ref{eq:multiplicity-ineq}) are certainly important results, they
address unconventionally defined critical points.  Both the
homological critical point, $T_h(G,H)$, and the end points of the
interval of possible non-analyticity, $T_c'(G)\le T_c(G)$, are defined
for sequences of Ising models without boundaries.  They are not
immediately related to the critical temperatures
$T_c^\mathrm{f}\le T_c^\mathrm{w}$ defined on related infinite systems
in terms of extremal Gibbs states with free/wired boundary conditions.

To bound these critical temperatures, consider a sequence of pairs of
weakly dual Ising models which satisfy the conditions of Theorems
\ref{th:lts-convergence} and \ref{th:hts-convergence-homo} with the
asymptotic rate $R>0$, with an additional assumption that matrices
$G_t$ and $H_t$ are incidence matrices of graphs, that is, they have
uniform column weights $\ell=\ell^*=2$.  In addition we assume that
the graph sequences converge weakly to a pair of infinite transitive
graphs, which we denote $\mathcal{G}=(\mathcal{V},\mathcal{E})$ and
$\mathcal{H}=(\mathcal{F},\mathcal{E})$, where $\mathcal{F}$ is the
set of faces in $\mathcal{G}$.  Weak convergence is defined as
follows: for some chosen vertex $0\in\mathcal{V}$, there is an
increasing sequence $\rho_t\in\mathbb{N}$ such that a ball
$\mathcal{B}(0,\rho_t)\subset\mathcal{G}$ of radius $\rho_t$ centered
at $0$, is isomorphic to a ball in $\mathcal{G}_t$.

These conditions necessarily imply that matrices $G_t$ and $H_t$
describe mutually-dual locally-planar graphs, and also that the graphs
$\mathcal{G}$ and $\mathcal{H}$ are mutually dual.

 Examples of such a sequence are given by sequences of finite
hyperbolic graphs constructed\cite{Siran-2001,Sausset-Tarjus-2007} as
finite quotients of the regular $\{f,d\}$ tilings of the infinite
hyperbolic plane, ${\bf H}(f,d)$, with $df/(d+f)>2$.  A graph in
such a sequence gives a tiling of certain surface, with $d$ regular
$f$-gons meeting in each vertex.  Hyperbolic graphs have been
extensively discussed in relation to quantum error correcting
codes\cite{Delfosse-Zemor-2010,%
  Delfosse-Zemor-2012,%
  Delfosse-2013,Delfosse-Zemor-2014,Breuckmann-Terhal-2015,%
  Breuckmann-Vuillot-Campbell-Krishna-Terhal-2017,Breuckmann-thesis-2017}.
Given such a finite locally-planar transitive graph with $n$ edges,
the quantum CSS code is a surface
code\cite{kitaev-anyons,Dennis-Kitaev-Landahl-Preskill-2002}; it is
constructed from the vertex-edge and plaquette-edge incidence
matrices, $G$ and $H$ respectively.  Here $H$ is also a vertex-edge
incidence matrix of a dual graph, which corresponds to the dual tiling
$\{d,f\}$ of the same surface.  Such a code has the minimal distance
scaling logarithmically with $n$, and it encodes $k=2g=2+nR$ qubits
into $n$, where $g$ is the genus of the surface and $R=1-2/d-2/f$ is
the asymptotic rate.

An extremal Gibbs ensemble on any infinite locally planar transitive
graph can be characterized by the average magnetization $m$, the
asymptotic correlation decay exponent $\alpha$
[Eq.~(\ref{eq:asymptotic-alpha})], and a similarly defined  asymptotic
domain wall tension 
\begin{equation}
  \tau\equiv\tau(\mathcal{G};K)=
  \inf_{\{i,j\}\subset\mathcal{F}} 
  \tau_{\mathbf{e}(i,j)},\label{eq:asymptotic-tau} 
\end{equation}
where $\mathbf{e}(i,j)$ is a defect that connects a pair of frustrated
plaquettes $i$ and $j$.  Generally, $\alpha=0$ whenever spontaneous
magnetization $m$ is non-zero.  A non-zero magnetization on a locally
planar transitive graph also implies $\tau>0$.  [This is a
generalization of the result from
Ref.~\onlinecite{Lebowitz-Pfister-1981}, see the proof in
Appendix~\ref{app:tension-lower-bound}.]  Respectively,
electro-magnetic duality (\ref{eq:em-duality-tension}) implies
\begin{statement}
 \label{th:planar-duality}
 Let $\mathcal{G}$ and $\mathcal{H}$ be a pair of infinite mutually
 dual locally-planar transitive graphs.  Denote
 $T_c^\mathrm{f}(\mathcal{G})$ and $T_c^\mathrm{w}(\mathcal{H})$ the
 critical temperatures of the extremal Gibbs ensembles for Ising
 models on $\mathcal{G}$ and $\mathcal{H}$ with free and wired
 boundary conditions, respectively.  Then these temperature are
 Kramers-Wannier duals of each other,
 \begin{equation}
   \label{eq:infinite-graph-duality}
   T_c^\mathrm{f}(\mathcal{G})=[T_c^\mathrm{w}(\mathcal{H})]^*. 
 \end{equation}
 For each model, in the ordered phase, $T<T_c$, $\alpha=0$ and
 $\tau>0$, while in the disordered phase, $T>T_c$, $\alpha>0$ and
 $\tau=0$.
\end{statement}
We can now  prove the following:
\begin{theorem}
  \label{th:wired-homo-inequality}
  For any regular $\{f,d\}$ tiling of an infinite hyperbolic plane,
  $fd/(f+d)>2$, the critical temperatures of the Ising model with free
  and wired boundary conditions,
  $T_c^\mathrm{f}=1/K_c^\mathrm{f}$ and
  $T_c^\mathrm{w}=1/K_c^\mathrm{w}$, satisfy 
  \begin{equation}
    \label{eq:Kc}
    K_c^\mathrm{f}-K_c^\mathrm{w}\ge R\ln2,\quad R=1-2/f -2/d.    
  \end{equation}
\end{theorem}
\begin{proof}
  For any regular $\{f,d\}$ tiling ${\cal G}\equiv {\bf H}(f,d)$ of the hyperbolic
  plane, consider a sequence of finite mutually dual locally planar
  transitive graphs $\mathcal{G}_t$ and ${\cal H}_t$, where the
  sequence ${\cal G}_t$ weakly converges to ${\cal G}$.  The
  corresponding sequence of incidence matrices satisfies the
  conditions of Theorems \ref{th:lts-convergence} and
  \ref{th:hts-convergence-homo} with the asymptotic rate $R>0$.
  Transitivity implies that the free energy density converges in a
  finite circle around $K=0$, see Remark \ref{th:hts-analytic-seq}-2.
  While we are not sure of convergence for larger $K$, Lemma
  \ref{th:subsequence-lemma} guarantees the existence of a subsequence
  of graphs, and corresponding pairs of incidence matrices $G_t$,
  $H_t$, $t\in\mathbb{N}$, such that the sequences of free energy
  densities $f(G_t;K)$ and $f(H_t;K)$ converge.  For such a sequence,
  the specific homological difference $\Delta f(G_t,H_t;K)$ also
  converges, which guarantees $\Delta f<R\ln2$ outside of the dual
  homological phase, $K>K_h^*(H,G)$.  Such an inequality implies the
  existence of an $\epsilon>0$ such that
  $\Delta f(G_t,H_t;K)<R\ln2-\epsilon/2$ at all sufficiently large
  $t$.  In turn, Eq.~(\ref{eq:average-tension}) implies that the
  average defect tension is bounded away from zero,
  $\bar\tau_0(G_t)\ge \epsilon$.

  While defects that contribute to the average $\bar\tau_0(G_t)$ have
  large weight, we notice that the free energy increment
  (\ref{eq:defect-delta}) associated with an arbitrary defect is
  subadditive, see Eq.~(\ref{eq:delta-subadittivity}).  Thus, a
  large-weight defect can be separated into smaller pieces;
  subadditivity (\ref{eq:delta-subadittivity}) ensures that
  $\max(\tau_{\mathbf{e}_1},\tau_{\mathbf{e}_2})\ge
  \tau_{\mathbf{e}_1+\mathbf{e}_2}$ as long as
  $d_{\mathbf{e}_1+\mathbf{e}_2}=d_{\mathbf{e}_1}+d_{\mathbf{e}_2}$.
  Thus, if we start with a homological defect with the tension
  $\tau_\mathbf{c}\ge\epsilon >0$, at each division we can select a
  piece with the tension not smaller than $\epsilon$.  Moreover, since
  homological defects are cycles on the dual graph, we can first
  separate $\mathbf{c}$ into simple cycles of weight not smaller than
  the corresponding CSS distance which increases with $t$, and then
  cut such a cycle in half to obtain a defect compatible with the
  definition (\ref{eq:asymptotic-alpha}).

  Further, GKS inequalities imply that the tension is monotonously
  non-decreasing when individual bonds' coupling is increased.  Thus,
  for the same defect $\mathbf{e}$ on $\mathcal{G}_t$ and on the
  corresponding subgraph with wired boundary conditions,
  $\tau_\mathbf{e}(G_t^{\rm w};K)\ge \tau_\mathbf{e}(G_t;K)\ge \epsilon$;
  this inequality survives the infinite graph limit.  Transitivity of
  $\mathcal{G}$ ensures that for the defect
  $\mathbf{e}=\mathbf{e}(i,j)$ connecting frustrated plaquettes $i$
  and $j$, the defect tension $\tau_\mathbf{e}\ge\epsilon$ depends
  only on the distance $d_\mathbf{e}$, which is the distance between
  vertices $i$ and $j$ on the dual graph $\mathcal{H}$.  This proves
  $\tau\ge\epsilon>0$ for the Ising model with wired boundary
  conditions on graph $\mathcal{G}$, at temperatures below the dual
  homological phase, $K>K_h^*(H,G)$.  Thus,
  $K_c^{\rm w}(\mathcal{G})\le K_h^*(H,G)$.

  Duality (\ref{eq:infinite-graph-duality}) also ensures that
  $K_c^\mathrm{f}(\mathcal{H})\ge K_h(G,H)$; inequality
  (\ref{eq:ineq-three}) gives Eq.~(\ref{eq:Kc}).
\end{proof}

\textsc{Remarks:} \ref{th:wired-homo-inequality}-1.  An interesting
fact about systems with finite rates $R>0$ is that electro-magnetic
duality (\ref{eq:em-duality-tension}) does not guarantee that area-law
exponent $\alpha_\mathbf{m}(G;K)$ be zero at low temperatures.  While
``area'' is the defect distance $d_\mathbf{m}$, the smallest number of
bonds in an equivalent defect, the ``perimeter'' is the
number of spins involved in the product, the syndrome weight
$\wgt \mathbf{s}$, where $\mathbf{s}=\mathbf{m}G^T$.  Standard
area/perimeter law argument assumes that perimeter can be
parametrically smaller than the area; this is not necessarily true for
systems with non-amenable Tanner graphs.

\noindent\ref{th:wired-homo-inequality}-2.  Even in the case of a
pair of locally planar graphs, a linear domain wall $\mathbf{e}$
connecting a pair of frustrated plaquettes may have a large perimeter
in the dual model, because of the additional spins corresponding to
the homological defects.  Any such defect that crosses the domain wall
(changes the sign of the corresponding spin average) increases the
perimeter in the dual model.  Such additional defects are absent with
free boundary conditions as considered in Theorem
\ref{th:wired-homo-inequality}. 

\subsubsection{Numerical results}

In addition to analytical bounds presented above, we also analyzed
numerically Ising models on several finite transitive hyperbolic
graphs constructed\cite{Siran-2001,Sausset-Tarjus-2007} as finite
quotients of the regular $\{5,5\}$ tilings of the infinite hyperbolic
plane.  We used canonical ensemble simulations with both local
Metropolis updates\cite{Metropolis-etal-1953} and Wolff cluster
algorithm\cite{Wolff-1989}, to compute the average magnetization
$m=\langle M\rangle/N$, susceptibility
$\chi=(\langle M^2\rangle-\langle M\rangle^2)/NT^2$, average energy
per bond $\varepsilon\equiv\langle E\rangle/n$, specific heat
$C=(\langle E^2\rangle-\langle E\rangle^2)/nT^2$, and the fourth
Binder cumulant\cite{Binder-1981}
$U_4=1-\langle S^4\rangle/(3\langle S^2\rangle^2)$. Here
$ M=|\sum_i S_i|$ is the (magnitude of the) total magnetization,
$E=- \sum_{\langle ij\rangle} S_i S_j$ is the total energy, $N$ and
$n$ respectively denote the number of spins and bonds, and
$\langle \cdot\rangle$ denotes the ensemble average.  For Metropolis
simulations, each run consisted of 128 cooling-heating cycles, with
1024 full graph sweeps at each temperature, with additional averaging
over 64 independent runs of the program.  The number of sweeps at each
temperature was sufficient to make any hysteresis unnoticeable.  For
Wolff algorithm simulations, each run consisted of 16 cooling-heating
cycles, with 4096 cluster updates at each temperature, and additional
averaging over 64 independent runs of the program.  The resulting
averages are shown in Figs.~\ref{fig:M} to \ref{fig:X}, where lines
(dots) show the data obtained with cluster (local Metropolis) updates,
respectively.  The results obtained using the two methods are very
close.

\begin{figure}[htbp]
  \centering
  \includegraphics[width=1.\columnwidth]{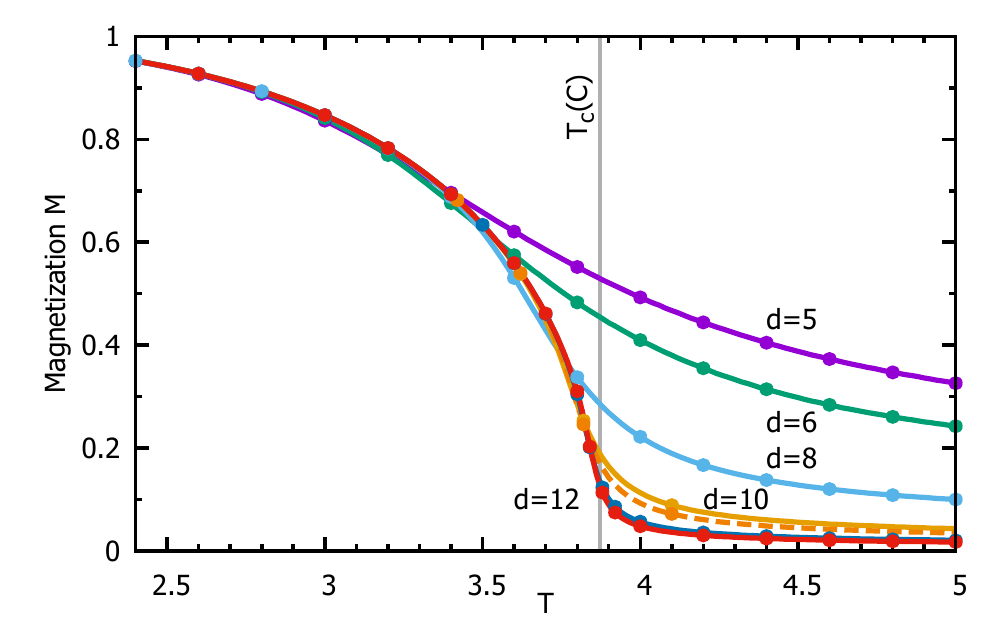}\\
  \includegraphics[width=\columnwidth]{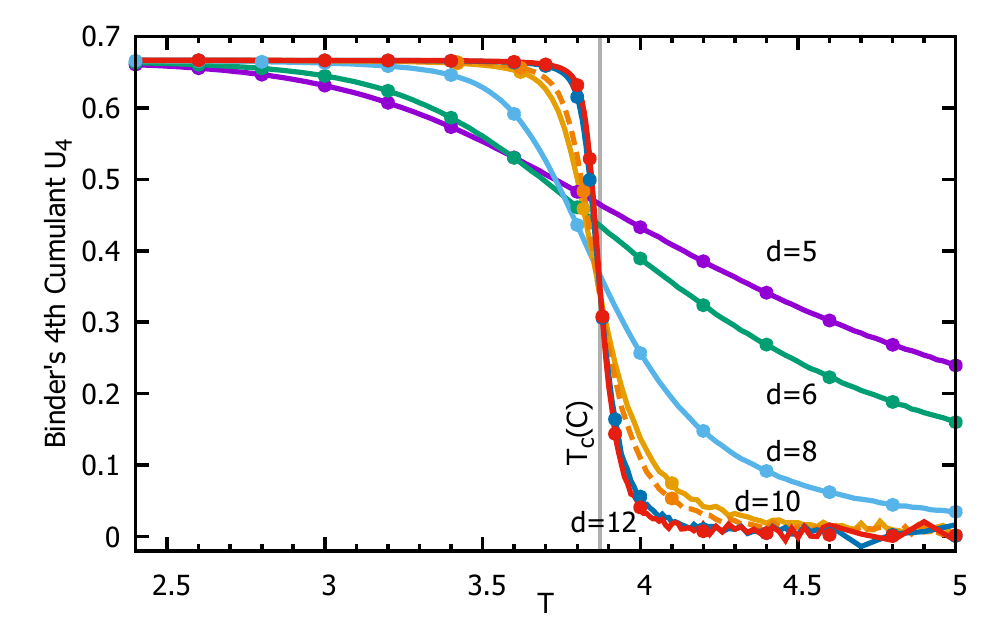}
  \caption{(Color online) Average magnetization (top) and Binder's
    fourth cumulant (bottom) as a function of temperature, for
    transitive graphs listed in Tab.~\ref{tab:graphs} with minimal
    distances as indicated.  Dashed lines show the data for the larger
    $d=10$ graph.  Lines show the data obtained using cluster updates;
    points show the data from simulations using local Metropolis
    updates.  Vertical line shows the critical temperature $T_c(C)$
    extrapolated from the positions of the specific heat maxima, see
    Fig.~\ref{fig:fit}.  While both sets of data do cross near
    $T_c(C)$, there is significant drift with increased graph size.
    In addition, the curves are near parallel which makes reliable
    extraction of the critical temperature difficult.  }
  \label{fig:M}
\end{figure}
\begin{figure}[htbp]
  \centering
  \includegraphics[width=\columnwidth]{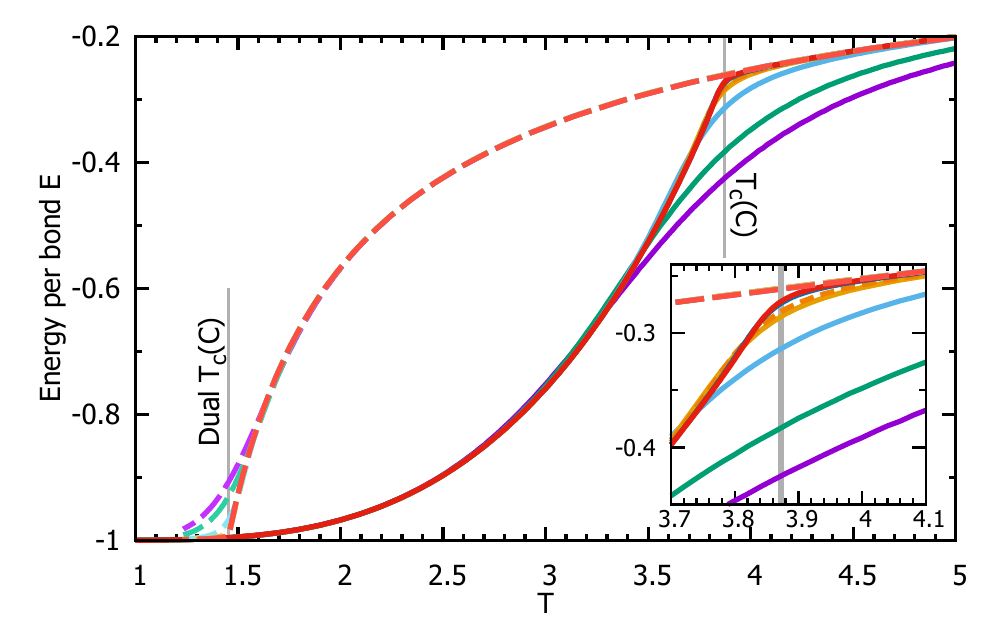}
  \caption{(Color online).  Solid lines: energy per bond from Wolff
    cluster calculations as a function of temperature, as in
    Fig.~\ref{fig:M}.  These data are converted with the help of the
    exact duality (\ref{eq:duality}) to give energies in the dual
    model (long dashes).  With increasing graph sizes, the difference
    between the original and dual energies decreases above the
    empirically found $T_c(C)$ (Fig.~\ref{fig:fit}) and below the
    corresponding Kramers-Wannier dual, $T_c^*(C)$.  Inset: close up
    of the plots near $T_c(C)$.}
  \label{fig:energy}
\end{figure}

\begin{figure}[htbp]
  \centering
  \includegraphics[width=\columnwidth]{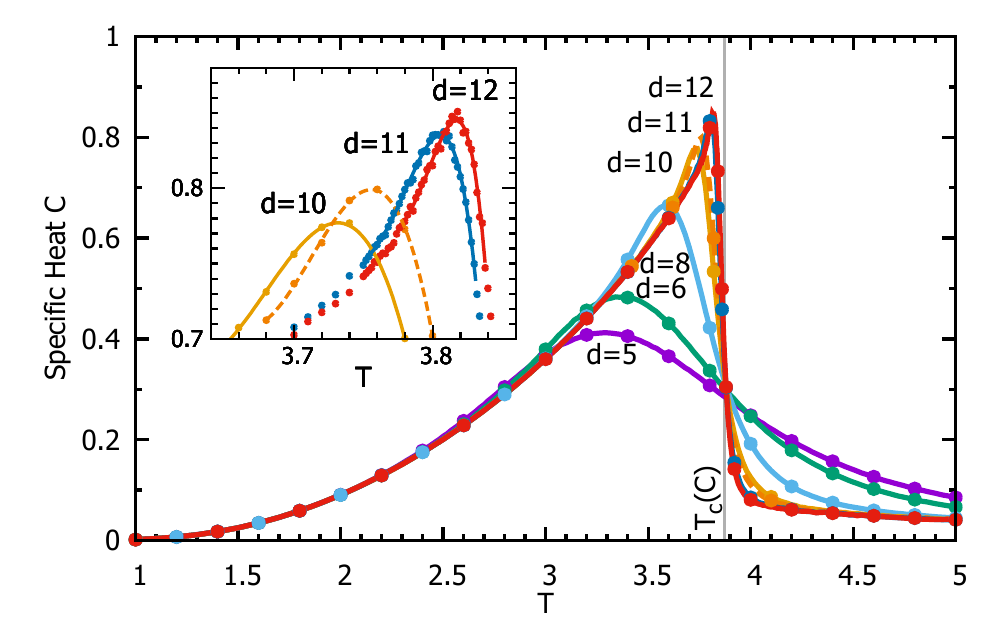}
  \caption{(Color online). As in Fig.\ \ref{fig:M} but for the
    specific heat.  Inset: fitting for maxima.  Data points in the
    inset are from the Wolff cluster calculations, while the lines are
    obtained using non-linear fits with general quartic polynomials of
    the form $y=y_m +a_2(x-x_m)^2+\ldots+ a_4(x-x_m)^4$, which give
    the coordinates of the maximum $(x_m,y_m)$ nearly independent from
    the rest of the coefficients.}
  \label{fig:C}
\end{figure}
\begin{figure}[htbp]
  \centering
  \includegraphics[width=\columnwidth]{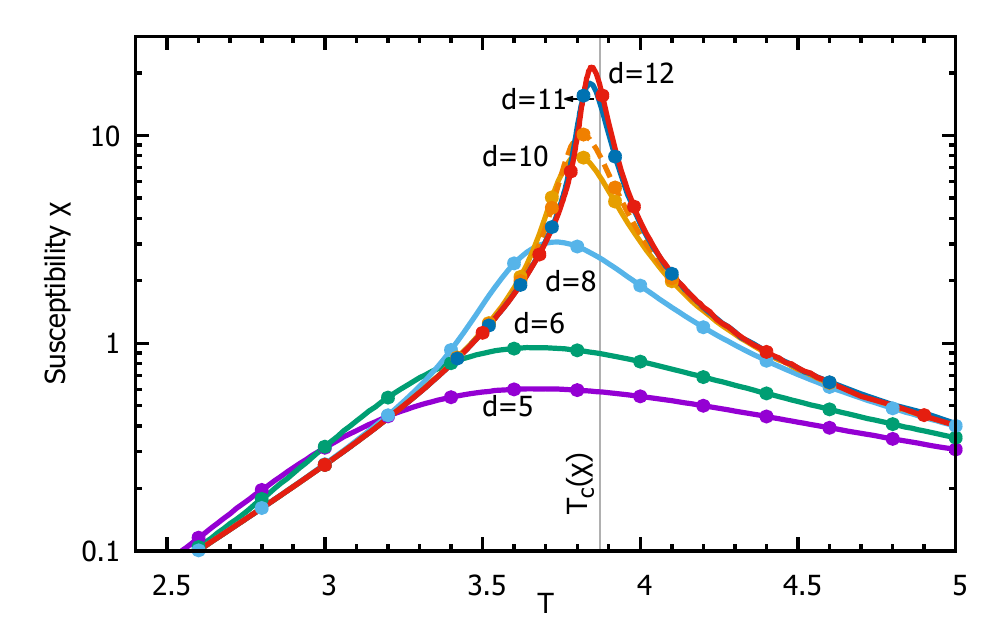}
  \caption{(Color online).  As in Fig.\ \ref{fig:M} but for the
    susceptibility $\chi(T)$, plotted in semi logarithmic scale.  The
    vertical line shows the critical temperature extrapolated from the
    susceptibility maxima, see Fig.~\ref{fig:fit}.}
  \label{fig:X}
\end{figure}

The parameters of the graphs used in the simulations are listed in
Tab.~\ref{tab:graphs}.  The first three graphs we obtained from
N.~P.~Breuckmann\cite{Breuckmann-thesis-2017}.  We generated the
remaining graphs with a custom \texttt{gap}\cite{GAP4} program, which
constructs coset tables of freely presented groups obtained from the
infinite van Dyck group $D(5,5,2)=\langle a,b|a^5,b^5,(ab)^2\rangle$
[here $a$ and $b$ are group generators, while the remaining arguments
are \emph{relators} which corresponds to imposed conditions,
$a^5=b^5=(ab)^2=1$] by adding one or more relator obtained as a pseudo
random string of generators, until a finite group is obtained.  Given
such a finite group ${\cal D}$, the vertices, edges, and faces are
enumerated by the right cosets with respect to the subgroups
$\langle a\rangle$, $\langle ab\rangle$, and $\langle b\rangle$,
respectively.  The vertex-edge and face-edge incidence matrices $G$
and $H$ are obtained from the coset tables.  Namely, non-zero matrix
elements are in positions where the corresponding pair of cosets share
an element.  Finally, the distance $d$ of the CSS code ${\cal Q}(G,H)$
was computed using the random window algorithm, which has the
advantage of being extremely fast when distance is
small\cite{Dumer-Kovalev-Pryadko-2014,Dumer-Kovalev-Pryadko-IEEE-2017}.
With the exception of the graph with $n=7440$, the graphs used have
the smallest size for the given distance.

\begin{table}[htbp]
  \centering
  \begin{tabular}[c]{c|c|c|c}
    vertices $r$ & edges $n$ & homology rank $k$ & CSS distance $d$ \\ \hline 
    32&80&18&5\\ 60&150&32&6\\360&900&182&8\\ 
    1920&4800&962&10\\ 2976&7440&1490&10\\ 
    8640&21600&4322&11\\ 12180&30450&6092&12
  \end{tabular}
  \caption{Parameters of the graphs used in the simulations.}
  \label{tab:graphs}
\end{table}

The obtained plots of magnetization and Binder's fourth cumulant are
shown in Fig.~\ref{fig:M}; the corresponding curves on largest graphs
are nearly indistinguishable, consistent with convergence at large
$n$.  We note that the crossing point in the Binder's fourth cumulant
show a significant drift with the system size, see lower plot on
Fig.~\ref{fig:M}.  This is not surprising, given that the original
scaling analysis\cite{Binder-1981} only applies to locally flat
systems, whereas the hyperbolic graphs have a uniform negative
curvature.  On both plots, the curves for larger system sizes are near
parallel to each other, which makes the identification of the phase
transition point from the corresponding crossing points difficult.

Fig.\ \ref{fig:energy} shows energy per bond as a function of
temperature.  To illustrate the properties of the specific homological
difference, see Theorems \ref{th:lts-convergence} and
\ref{th:hts-convergence-homo}, we also plot the energy per bond of the
exact dual models obtained from the same data using
$\varepsilon^*(K^*)=-\sinh(2K)\, \varepsilon(K)-\cosh(2K)$, derived
from Eq.~(\ref{eq:duality}).  The plot shows that as the size of the
graph increases, the difference between the energies
$\varepsilon^*(T)$ and $\varepsilon(T)$ decreases with increasing
graph size both above $T_c(C)$ and below the corresponding
Kramers-Wannier dual, $T_c^*(C)$, while a finite difference remains
for the intermediate temperatures.  This is consistent with the
identification $T_h^*=T_c(C)$.

The plots for specific heat $C(T)$ (Fig.~\ref{fig:C}) and magnetic
susceptibility $\chi(T)$ (Fig.~\ref{fig:X}) show well developed maxima
which become sharper and higher with increasing system sises.  Notice
that a unique point of divergence of the specific heat necessarily
coincides with the dual homological temperature $T_h^*$.

We obtained the positions of the specific heat and magnetic
susceptibility maxima by fitting the data in the vicinity of the
corresponding maxima with quartic polynomials as explained in the
caption of Fig.~\ref{fig:C}.  The resulting positions of the maxima
are plotted in Fig.~\ref{fig:fit} as a function of $x=1/n^{1/2}$.  The
error bars of the positions of the maxima have errors in the third
digit; the observed minor scattering of the data is a feature of the
corresponding graphs.

\begin{figure}[htbp]
  \centering
  \includegraphics[width=\columnwidth]{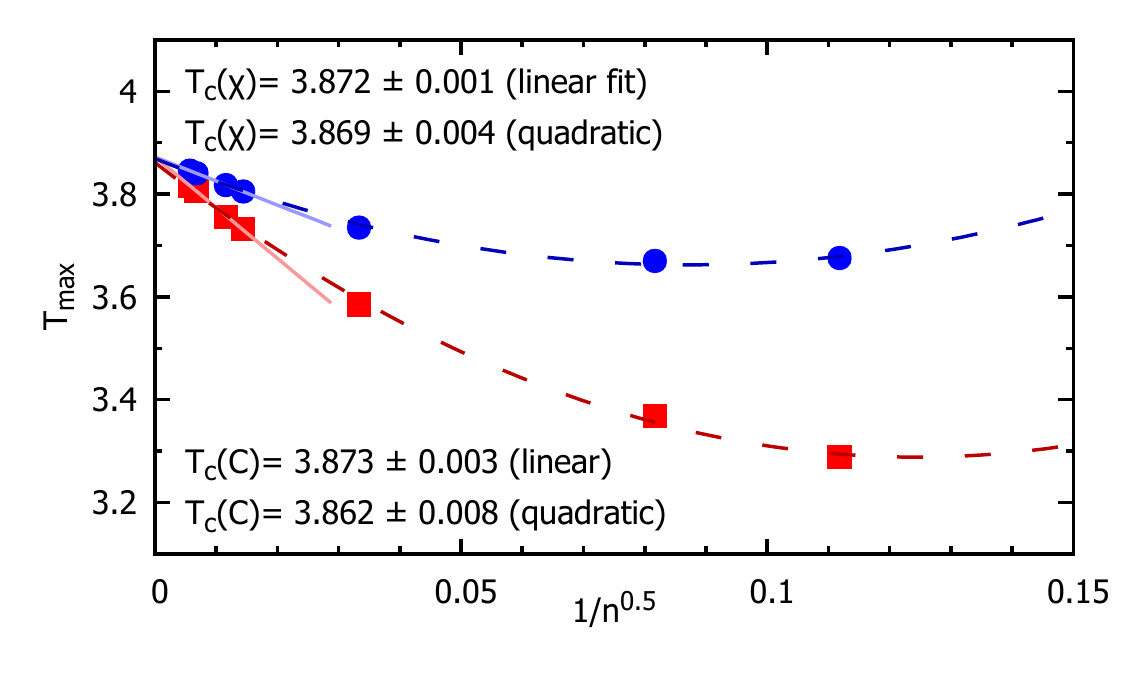}
  \caption{(Color online). Extrapolation of the specific heat and
    susceptibility maxima to infinite system size.  Red squares (blue
    circles) show the positions of the specific heat (susceptibility)
    maxima extracted from the data on Figs.~\ref{fig:C} and
    \ref{fig:X}, respectively for graphs of different size, plotted as
    a function of $1/n^{1/2}$, where $n$ is the number of edges in the
    corresponding graphs, see Tab.~\ref{tab:graphs}.  Solid (dashed)
    lines are obtained as linear (quadratic) fits to the data, where
    only the four leftmost points were used for the linear fits.  This
    results in the extrapolated critical temperature values as
    indicated.}
  \label{fig:fit}
\end{figure}

While the size dependence is not monotonic in the case of
susceptibility maxima, the data points for larger graphs show
approximately linear dependence on $x$.  Linear extrapolation to
infinite size ($x=0$) gives $T_c\approx 3.872\pm0.003$ for both sets
of data.  This value is consistent with the lower bound
(\ref{eq:multiplicity-ineq}) for the infinite graph with wired
boundary conditions, which gives in the present case $T_c\ge 2.668$.
In comparison, the transition for a square-lattice Ising model is in
the self-dual point, $T_\mathrm{s.d.}=2/\ln(1+\sqrt2)\approx 2.269$.

We note that even though we expect Ising model on hyperbolic graphs to
have mean field criticality, conventional finite size scaling theory
does not apply here.  In particular, this is seen from the absence of
the well defined crossing point in the data for Binder's fourth
cumulant, see the lower plot on Fig.~\ref{fig:M}.  Therefore, we had
to experiment on how to extrapolate the positions of the maxima to
estimate the critical temperature.  The scaling with $x=1/n^{1/2}$ was
chosen since it gives near identical estimates for the critical
temperatures from the maxima of $C(T)$ and $\chi(T)$, cut off at
different maximum sizes (we tried $d_\mathrm{max}=8$ and above).

We also note that the data shows good convergence with increased
system size, without the need for the subsequence construction
described in Sec.~\ref{sec:analyticity}.

\section{Discussion}
\label{sec:discussion}

\subsection{Summary of the results}

We considered pairs of weakly dual Ising models with few-body
couplings, defined via sequences of degree-limited bipartite coupling
graphs, with the focus on the case where the rank $k$ of the first
homology group of the corresponding two-chain complex scales
extensively with the system size.  This construction is needed to
avoid introducing the boundaries, which are known to affect the
position of the critical point in non-amenable graphs, and also to
connect to applications, e.g., in quantum information theory, where
results for large but finite systems are of interest.  Here, extensive
scaling of $k$ corresponds to quantum error correcting codes with
finite rates $R>0$.  Important examples include two-body Ising models
on families of finite transitive hyperbolic graphs which weakly
converge to regular $\{f,d\}$-tilings of the hyperbolic plane with
$df/(d+f)>2$; the corresponding limiting rates $R=1-2/d-2/f$ are
non-zero.

Our main result is Theorem \ref{th:lts-convergence}, which guarantees
the existence of a low-temperature, low-disorder region where
homological defects are frozen out---in the thermodynamical limit they
have no effect on the free energy density.  Duality guarantees the
existence of a high-temperature phase where extensive homological
defects have near zero free energy cost, see Theorem
\ref{th:hts-convergence-homo}.  At all temperatures below this phase,
the average defect tension is non-zero, see
Eq.~(\ref{eq:average-tension}).

With the help of duality and a known bound on high-order cumulants, we
established the absolute convergence of both the high- and
low-temperature series expansions of the free energy density in finite
regions which include vicinities of the real temperature axis around
the zero and infinite temperatures, respectively.  We used a
subsequence construction to ensure the convergence of free energy
density at all temperatures, and defined the critical temperatures as
the real-axis points of non-analyticity of the limiting free energy
density.  For these critical temperatures, we derived several
inequalities, in particular, an analog of multiplicity of the critical
points, which guarantees that with $R>0$, critical point of the free
energy density is affected by the summation over the topological
defects.

As an application of obtained bounds, we proved the multiplicity of
phase transitions on all regular tilings $\mathbf{H}(f,d)$ of the infinite
hyperbolic plane, $df/(d+f)>2$.

We also simulated the phase transition on a sequence of self-dual
$\{5,5\}$ transitive hyperbolic graphs without boundaries, with up to
$n_\mathrm{max}=30450$ bonds numerically.  Our data shows good
convergence with increasing system sizes, with a single specific heat
maximum which sharpens with the increasing system size.  If the
corresponding position $T_c(C)\approx 3.872\pm0.003$ is the only
singularity of the free energy, then necessarily it coincides with the
dual homological point, $T_h^*=T_c(C)$.

\subsection{Open questions}

\noindent{\bf 1.}\ The rightmost point of the homological region
established in Theorem \ref{th:lts-convergence} on the $p$-$T$ plane
has the same value $p_\mathrm{max}$ as can be also obtained using the
energy-based arguments\cite{Dumer-Kovalev-Pryadko-bnd-2015}, which
apply at $T=0$.  Either of these results also
implies\cite{Nishimori-book,Kovalev-Pryadko-SG-2015} that the portion
of the Nishimori line at $p<p_\mathrm{max}$ is in the homological
region.  It should be possible to establish the existence of a
homological region in the intermediate temperature points, but we
could not find the corresponding arguments.

\noindent{\bf 2.}\ The proof of Statement \ref{th:hts-convergence} is
based on overly generic bounds\cite{Feray-Meliot-Nikeghbali-2013} for
cumulants of a sum of random variables with a given dependency graph.
In the case of the Ising model, it should be possible to construct a
stronger lower bound for absolute convergence of the HTS.  We expect
that the same bound as in Theorem \ref{th:hts-convergence-homo} should
apply.  Such a bound would be consistent with that from
high-temperature series expansions for spin
correlations\cite{Fisher-upper-1967}, and it would also be consistent
with the analysis of the higher-order derivatives of free
energy\cite{Lebowitz-1972}, as well as the naive expectation that
$T_c(G)=T_h^*(H,G)$.

\noindent{\bf 3.}\ In addition to the case in Remark
\ref{th:hts-analytic-seq}-2, the infinite subsequence construction of
Corollary \ref{th:hts-analytic-seq} is also not needed when the
sequence of Tanner graphs has a well defined distributional limit
(Benjamini-Schramm or ``left''
convergence\cite{Benjamini-Schramm-2001,%
  Borgs-Chayes-Kahn-Lovasz-2013,Lovasz-2016}).  Important examples are
given by the Tanner graphs of hypergraph-product and related
codes\cite{Tillich-Zemor-2009,Kovalev-Pryadko-Hyperbicycle-2013} based
on specific families of sparse random matrices.  For such sequences,
it would be nice to establish the conditions for convergence of the
free energy density or spin averages for all $K>0$, to supercede the
subsequence construction of Lemma \ref{th:subsequence-lemma}.

\begin{acknowledgments}
  LPP is grateful to N. P. Breuckmann for the explanation of the
  quotient group construction of the hyperbolic graphs.  This work was
  supported in part by the NSF under Grants No.\ PHY-1415600 (AAK),
  ECCS-1102074 (ID), and PHY-1416578 (LPP).  LPP also acknowledges
  hospitality by the Institute for Quantum Information and Matter, an
  NSF Physics Frontiers Center with support of the Gordon and Betty
  Moore Foundation.
\end{acknowledgments}

\appendix 

\section{Proof of Theorem \ref{th:lts-convergence}}
\label{app:ths-convergence-proof}
\ltsconvergence*

The statement of the theorem immediately follows from the following
technical Lemma, see the proof in
Ref.~\onlinecite{Kovalev-Prabhakar-Dumer-Pryadko-2018}
\begin{lemma}
  \label{th:homological-bound}  
  Consider a pair of Ising models defined in terms of weight-limited
  matrices $G$ and $H$ with orthogonal rows, such that the matrix $H$
  has a maximum row weight $m$.  Let $d_G$ denote the CSS distance
  (\ref{eq:CSS-dist}), the minimum weight of a frustration-free
  homologically non-trivial defect
  $\mathbf{c}\in{\cal C}_H^\perp\setminus{\cal C}_G$.  Denote
  $S\equiv e^{-2K}(1-p)+e^{2K}p$, and assume that 
  $ (m-1)\, S<1$.  Then, the average homological difference
  (\ref{eq:homo-difference}) satisfies
  \begin{equation}
[\Delta f(G,H;K)]_p\le {(m-1)^{d_G} S^{d_G+1}\over
    1-(m-1)\,S} .\label{eq:homological-bound}
\end{equation}
\end{lemma}

\section{Proof of inequalities in Sec.~\ref{sec:general}}

\noindent(\textbf{i}) 
The proof of the monotonicity of the
homological difference (in the absence of flipped bonds),
$$ 
{d\over dK} \Delta
f_\mathbf{0}(G,H;K)\le0,\eqno(\ref{eq:monotonicity}),
$$
is similar to the proof\cite{Bricmont-Lebowitz-Pfister-1980} of the
monotonicity of the tension.  We combine the logarithms in
Eq.~(\ref{eq:homo-difference}), decompose $Z_\mathbf{e}(H^*;K)$ as a
sum of $Z_\mathbf{c}(G;K)$ over non-equivalent codewords $\mathbf{c}$,
and write
$$
{d\over
  dK}{Z_\mathbf{c}(G;K)\over
  Z_\mathbf{0}(G;K)}={Z_\mathbf{c}(G;K)\over
  Z_\mathbf{0}(G;K)}\sum_{b\in\mathcal{B}}\left(\langle
  R_b\rangle_\mathbf{c}-
\langle R_b\rangle_\mathbf{0}\right)\le 0.
$$
The desired inequality (\ref{eq:monotonicity}) follows from the
monotonicity of the logarithm. \smallskip

\noindent(\textbf{ii})
The first  inequality in 
\begin{displaymath}
  | \tau_\mathbf{c,e}|\le\tau_\mathbf{c,0}\le2K
  \eqno(\ref{eq:tension-generic-bounds})  
\end{displaymath}
follows from the second GKS
inequality\cite{Griffiths-1967,Kelly-Sherman-1968} applied in the dual
system [where, according to electric-magnetic duality, the defect
becomes an average of the corresponding product of spins, see
Eq.~(\ref{eq:em-duality})].  Depending on the sign of
$\tau_\textbf{c,e}$, duality gives
$\langle R_\mathbf{c+e}\rangle\ge \langle
R_\mathbf{e}\rangle \langle R_\mathbf{c}\rangle$ or
$\langle R_\mathbf{e}\rangle\ge \langle R_\mathbf{c}\rangle \langle
R_\mathbf{e+\mathbf{c}}\rangle$, where $R_\mathbf{e}$ is the product
of bonds corresponding to non-zero bits in the binary vector
$\mathbf{e}$.  The second inequality, in a more general form,
$\tau_\mathbf{e}\equiv\tau_\mathbf{e,0}\le 2K$, follows from the Gibbs
inequality 
$$
F_\mathbf{e}(G;K)-F_\mathbf{0}(G;K)\le 2K\sum_{b:e_b\neq0} \langle
R_b\rangle_{G;K}\le 2K \wgt(\mathbf{e}),
$$
if we take a  minimal-weight vector equivalent to $\mathbf{e}$, in
which case $\wgt(\mathbf{e})=d_\mathbf{e}$.

\noindent(\textbf{iii}) To prove the lower bound on the average
tension,
$$
\zeta \bar \tau_p
  \ge R\ln 2-[\Delta
  f_\mathbf{e}]_p,\eqno(\ref{eq:average-tension})
$$
we first define the constant $\zeta$ as the average minimum weight of
all $2^k$ codewords divided by the code length $n$,
\begin{equation}
\zeta=(2^kn)^{-1}\sum_{\mathbf{c}}d_\mathbf{c}.\label{eq:zeta}
\end{equation}
An upper bound on
$\zeta$ can be obtained if we take the codewords $\mathbf{c}$ as
linear combinations of $k$ inequivalent codewords $\mathbf{c}_i$,
$i\in\{1,\ldots,k\}$ (it is likely that smaller-weight equivalent
codewords can be found).  In this case the codewords form a binary 
code, and the average weight is exactly a half of the length $n'$ of
the code\cite{MS-book}, where $n'=\left|\cup_{i=1}^k I({\bf c}_i)\right|$ is
the weight of the union of the supports of the basis codewords.
Clearly, $n'\le n$, which gives $\zeta\le 1/2$.  Combining with a
lower bound on the weight of non-trivial codewords,
$d_\mathbf{c}\ge d_{\cal G}$, $\mathbf{c}\not\simeq\mathbf{0}$, we
obtain
\begin{equation}
  \label{eq:zeta-constant}
  {1-2^{-k}\over n}d_G\le \zeta\le  {1\over2}. 
\end{equation}
We now proceed with deriving the inequality
(\ref{eq:average-tension}).  Start by expanding
$Z_\mathbf{e}(H^*;K)=\sum_\mathbf{c}Z_\mathbf{e+c}(G;K)$, where the
summation is over all $2^k$ mutually inequivalent codewords
$\mathbf{c}$.  Each of the terms with $\mathbf{c}\not\simeq\mathbf{0}$
can be written in terms of the corresponding tension
(\ref{eq:disorder-averaged-tension}),
$$
Z_\mathbf{e+c}(G;K)=e^{-\tau_\mathbf{c,e}(G;K)d_\mathbf{c}}Z_\mathbf{e}(G;K).
$$ 
Convexity of the exponent gives
\begin{eqnarray*}
{Z_\mathbf{e}(H^*;K)\over Z_\mathbf{e}(G;K)}&=&1+\sum_\mathbf{c\not\simeq0}
\exp(-\tau_\mathbf{c,e}d_c) \\ &\ge& 2^k
\exp\Biglb(-{2^{-k}}\sum_\mathbf{c}
\tau_{\bf c,e}d_\mathbf{c}\Bigrb),
\end{eqnarray*}
where for the trivial codeword $\mathbf{c}\simeq\mathbf{0}$ we set
$\tau_{\bf 0,e}d_\mathbf{0}=0$.  Taking the logarithm and rewriting
the sum over codewords in terms of the weighted average, with the help
of Eq.~(\ref{eq:zeta}) we obtain
$$
\Delta F_\mathbf{e}(G,H;K)\ge
k\ln2-\zeta n {\sum_\mathbf{c\not\simeq0} \tau_\mathbf{c,e}d_\mathbf{c}\over
  \sum_\mathbf{c\not\simeq0}d_\mathbf{c}}. 
$$
Eq.~(\ref{eq:average-tension}) trivially follows after averaging over
disorder and dividing by $n$.

\noindent(\textbf{iv})
The inequality 
$$
   K_h(G,H)-K_h^*(H,G)\ge R\ln2\eqno(\ref{eq:ineq-three})
$$
is based on the standard inequality for the derivative of the free
energy density, which is just the average energy per bond.  For the
case of homological difference we obtain, instead,
\begin{equation}
{d\over dK} \Delta f(G,H;K)={1\over n}\sum_{b\in\mathcal{B}}\left(\langle
R_b\rangle_{H^*;K}-\langle
R_b\rangle_{G;K}\right).\label{eq:energy-bound}
\end{equation}
The second term can be obtained from the first by freezing the spins
corresponding to homologically non-trivial defects; with the help of
GKS inequalities we obtain
$$
1\ge \langle
R_b\rangle_{G;K}\ge \langle
R_b\rangle_{H^*;K}\ge 0,
$$
which guarantees the derivative (\ref{eq:energy-bound}) to be between
$-1$ and $0$.  Integration gives the inequality
$$
\Delta f_t(K_2)-\Delta f_t(K_1)\le K_1-K_2,
$$
where $\Delta f_t(K)=\Delta f(G_t,H_t;K)$. 
We now take $K_1=K_h(G,H)$ and $K_2=K_h^*(H,G)$, so that in
the limit of the sequence, $\lim_t\Delta f_t(K_1)=0$ and $\lim_t\Delta
f_t(K_2)=R\ln2$.  Eq.~(\ref{eq:ineq-three})  trivially follows.

\section{Proof of Theorem \ref{th:hts-convergence-homo}}

\htsconvergencehomo*

\begin{proof}
  The proof is based on the special case of Theorem
  \ref{th:lts-convergence} in the absence of disorder, $p=0$, and the
  duality relation (\ref{eq:homological-diff-duality}), applied for
  each pair of matrices, $G_t$ and $H_t$, with $R_t=k_t/n_t$, and $K$
  replaced with its Kramers-Wannier dual, $K^*$.  The condition on $K$
  in Theorem \ref{th:lts-convergence} (with $G_t$ and $H_t$
  interchanged) becomes simply $(m-1)\tanh K<1$.  Convergence of
  sequences $\Delta f_\mathbf{0}(H_t,G_t;K^*)$ to $0$ and $R_t$ to $R$
  implies that of the sequence $\Delta f_\mathbf{0}(G_t,H_t;K)$ to
  $R\ln2$.
  \end{proof}

\section{Proof of Statement \ref{th:hts-convergence}}

The proof is based on Theorem 9.1.7 from
Ref.~\onlinecite{Feray-Meliot-Nikeghbali-2013}, which bounds cumulants
of a random variable $X$,
\begin{equation}
  \label{eq:cumulant}
\kappa_r(X)\equiv {d^r\over dt^r}\,\ln \mathbb{E}\left(e^{t
    X}\right)\biggr|_{t=0},  \;\,r\in\{0,1,\ldots\},
\end{equation}
where $X=\sum_{\alpha\in {\cal S}} Y_\alpha$ is a sum of random
variables with a given \emph{dependency graph}:
\begin{definition}
  A graph ${\cal D}$ with vertex set ${\cal S}$ is called a
  \emph{dependency graph} for the set of random variables
  $\{Y_\alpha,\alpha\in {\cal S}\}$ if for any two disjoint subsets
  ${\cal S}_1$ and ${\cal S}_2$ of ${\cal S}$, such that there are no
  edges in ${\cal D}$ connecting an element of ${\cal S}_1$ and an
  element of ${\cal S}_2$, the sets of random variables
  $\{Y_\alpha\}_{\alpha\in{\cal S}_1}$ and
  $\{Y_\alpha\}_{\alpha\in{\cal S}_2}$ are independent.
\end{definition}
The corresponding bound reads as follows:
\begin{lemma}[Theorem 9.1.7 from Ref.\
  \onlinecite{Feray-Meliot-Nikeghbali-2013}] 
  \label{th:cumulant-bound}
  Let $\{Y_\alpha\}_{\alpha\in{\cal S}}$ be a family of random
  variables with dependency graph ${\cal D}$.  Denote $N=|{\cal S}|$
  the number of vertices of ${\cal D}$ and $\Delta$ the maximal degree
  of ${\cal D}$.  Assume that the variables $Y_\alpha$ are uniformly
  bounded by a constant $A$.  Then, for the sum
  $X=\sum_{\alpha\in{\cal S}}Y_\alpha$, and for any $s\in\{0,1,\ldots\}$,
  one has
  \begin{equation}
    \label{eq:cumulant-bound}
    |\kappa_s(X)|\le 2^{s-1} s^{s-2} N (\Delta+1)^{s-1}A^s.
  \end{equation}  
\end{lemma}

\htsconvergence*

\begin{proof}[Proof of Statement \ref{th:hts-convergence} ] 
  The $s$-th coefficient of the HTS for the free energy
  $F(\Theta;K,h)$ is the scaled cumulant $-\kappa_s(X)/s!$, where
  $X=J\sum_{b\in{\cal B}}R_b+h' \sum_{v\in{\cal V}}S_v$.  Define the
  set of random variables $Y_\alpha$ as the union of the set of
  (scaled) spins $h S_v$ and bonds $K R_b$, then
  $|Y_\alpha|\le A\equiv \max(|h'|,|J|)$.  The corresponding
  dependency graph ${\cal D}$ can be obtained from the bipartite graph
  defined by the matrix $\Theta$ by connecting any pair of nodes for
  bonds which share the same spin.  In the original bipartite graph,
  each spin node has up to $\ell$ neighboring bond nodes, and each
  bond node has up to $m$ neighboring spin nodes.  In the modified
  graph, each bond node also connects with up to $(\ell-1)m$ bond
  nodes with common spins, which gives the total maximum degree of
  $\Delta=\ell m$.  We also have $N=|{\cal V}|+|{\cal E}|=r+n$,
  dividing by $n$ as appropriate for the free energy density we obtain
  the bound in part (a).  With $h=0$, we can drop the spin nodes from
  the dependency graph.  In this case the maximum degree is
  $\Delta'=(\ell-1)m$, which gives the result in part (\textbf{b}).
  Notice that in this case $N=n$, and the factor $C=(r/n+1)$ is
  replaced with $C'=1$.
\end{proof}

\section{Proof of Corollary \ref{th:hts-analytic-seq}.}

\htsanalyticseq*

\begin{proof}
  The result in Statement \ref{th:hts-convergence}(b) gives a uniform
  in $t$ bound on the coefficients of the HTS,
  \begin{eqnarray}\nonumber {|\kappa_s(\Theta_j)|\over s!}  
    &\le& 
          {2^{s-1} s^{s-2} (\Delta+1)^{s-1}J^s\over (2\pi
          s)^{1/2}(s/e)^s}\\ 
    &=&
        {1\over \sqrt{8\pi}(\Delta+1)}
        {[2e J\,(\Delta+1)]^s\over
        s^{5/2}},
        \label{eq:upper-bound-kappa-fe}
  \end{eqnarray} where $\Delta\equiv (\ell-1)m$ and we used the lower
  bound by Stirling, $r!\ge (2\pi 
  r)^{1/2}(r/e)^r$.  The bound (\ref{eq:upper-bound-kappa-fe}) is
  uniform in the sequence index $j\in\mathbb{N}$.  
  Thus one can select an infinite subsequence
  of ${\Theta}_j$, ${\Theta}_{j'(t)}$, $t\in\mathbb{N}$, where the
  function $j':\mathbb{N}\to\mathbb{N}$ is strictly increasing, so
  that the coefficients $\kappa_m(\Theta_{j'(t)})$ for $m=1$ converge
  with $t$.  Selecting an infinite subsequence of the one obtained
  previously to 
  ensure the convergence of the coefficients $\kappa_m$ for
  $m=2, 3,\ldots$, at each step we obtain an infinite subsequence such
  that all coefficients $\kappa_s$ with $s\le m$ converge with $t$.
  The statement in part 
  (\textbf{a}) is obtained in the limit of $m\to\infty$.  The uniform bound 
  (\ref{eq:upper-bound-kappa-fe}) also applies to the cumulants after
  we take the limit of the obtained subsequence,
  which implies absolute convergence (and thus analyticity of the
  limit) of the HTS 
  for free energy 
  density in the circle $|K|\equiv |\beta|J\le
  \{2e[(\ell-1)m+1]\}^{-1}$, which is exactly the statement  in part
  (\textbf{b}). 
\end{proof}

\section{Proof of Lemma \ref{th:subsequence-lemma}}
\label{app:subsequencelemma}

\subsequencelemma*

\begin{proof}
  For any $t$, the free energy density
  $f_t(K)=-n_t^{-1}\ln Z_\mathbf{0}(G_t,K)$ is bounded from both
  sides, 
  $$-M\le r_t\ln2/n_t-K\le f_t(K)\le r_t\ln2/n_t+K\le \ln2+M.$$  Therefore, we
  can use a subsequence construction to ensure convergence in any
  point $K\in I_M$.  Since the set of rational numbers $\mathbb{Q}$ is
  countable, we can repeat this construction sequentially on all
  rational points in $I_M$. The resulting infinite sequence $f_i(K)$
  converges in any rational point $K\in I_M\cap \mathbb{Q}$.  Further,
  the derivative of $f_i(K)$ is uniformly bounded,
  $-1\le f_i'(K)\le 0$.  Since the sequence converges on a dense
  subset of $I_M$, this guarantees the existence and the continuity of
  the limit in the entire interval.  Finally, each of $f_i(K)$ is
  concave and non-increasing; these properties survive the limit,
  although the resulting function may not necessarily be strictly
  concave.  Concavity guarantees the existence of one-sided
  derivatives.  The lower and upper bounds on these derivatives are
  inherited from those for $f_i'(K)$.
\end{proof}

\section{Proof of Theorem \ref{th:temperature-inequalities}}
\label{app:temperature-inequalities}

\temperatureinequalities*

\begin{proof} 
  There are three mutually exclusive possibilities: (a)
  $T_c'(G)<T_h(G,H)$, (b) $T_c'(G)>T_h(G,H)$, and (c)
  $T_c'(G)=T_h(G,H)$.  In the case (a), $T_c^*(H)=T_c'(G)$, since the
  functions $f_G(K)$ and $f_{H^*}(K)$ coincide in the homological
  region, i.e., for $K>K_h(G,H)$; Eq.~(\ref{eq:ineq-two}) is
  satisfied.  In the case (b), $T_c^*(H)=T_h(G,H)$, in order to
  recover the non-analyticity point for the homological difference;
  Eq.~(\ref{eq:ineq-two}) is saturated.  The goal of the Conditions is
  to deal with the case (c) which implies $T_c^*(H)\ge T_h(G,H)$; a
  strict inequality would violate Eq.~(\ref{eq:ineq-two}).  In the
  following we assume (c).

  Condition 1 implies that the (negative) curvature of $f_G(K)$ must
  diverge at $K_h=K_c'(G)$, which must be compensated by a divergent
  curvature of $f_{H^*}(K)$ in order to make $\Delta f(K)$ strictly
  convex in this point.  In this case $T_c^*(H)=T_c'(G)$; 
  Eq.~(\ref{eq:ineq-two}) is saturated.

  Condition 2 does the same, since divergent positive curvature of
  $\Delta f(K)$ at $K_h$ can only come from $f_{H^*}(K)$.

  Condition 3 is equivalent by duality to $T_c'(G)\ge T_c^*(H)$, which
  again gives Eq.~(\ref{eq:ineq-two}) since we assumed (c).
\end{proof}

\section{Proof of the lower bound for tension}
\label{app:tension-lower-bound}

On an infinite locally planar transitive graph $\mathcal{G}$, we would
like to prove the following bound for the asymptotic defect tension
(\ref{eq:asymptotic-tau}),
\begin{equation}
  \label{eq:tension-lower-bound-m}
  {d\tau(K)\over dK}\ge 2[m(K)]^2,
\end{equation}
the same inequality as has been previously proved on $\mathbb{Z}^D$ in
Ref.~\onlinecite{Lebowitz-Pfister-1981}.  This inequality is a trivial
consequence of the following Lemma, which gives a version of Eq.~(7)
from Ref.~\cite{Lebowitz-Pfister-1981} suitable to constructing a
bound for the defect tension defined by Eq.~(\ref{eq:defect-tension}).
\begin{restatable}{lemma}{tensionlowerbound}
  \label{th:tension-lower-bound}
  Let $\mathcal{G}=(\mathcal{V},\mathcal{E})$ be a finite transitive
  graph, and $G$ the corresponding vertex-edge incidence matrix with
  $n=|\mathcal{E}|$ columns.  Take a binary vector
  $\mathbf{e}\in\mathbb{F}_2^n$ selecting a set of edges
  $\mathcal{E}_\mathbf{e}\subset\mathcal{E}$ of size
  $|\mathcal{E}_\mathbf{e}|=\wgt(\mathbf{e})$, and a set of vertices
  $\mathcal{A}\subset \mathcal{V}$ of twice the size,
  $|\mathcal{A}|=2\wgt(\mathbf{e})$, such that the graph contains
  edge-disjoint paths connecting each edge to exactly two vertices in
  ${\cal A}$.  Then for the Ising model defined on the same graph, at
  any $K,h\ge0$, the free energy increment
  $\delta_\mathbf{e}(K)\equiv F_\mathbf{e}(G;K,h)-F_\mathbf{0}(G;K,h)$
  associated with the defect $\mathbf{e}$ satisfies
  \begin{equation}
  {d\delta_\mathbf{e}(K)\over dK}\ge  \langle S_i\rangle_\mathbf{0}
\sum_{v\in\mathcal{A}}  \langle S_v\rangle_\mathbf{e},
\label{eq:tension-lower-bound}  
\end{equation}
where the average $\langle S_v\rangle_\mathbf{e}$ is calculated in the
presence of the defect $\mathbf{e}$; by transitivity
$\langle S_i\rangle_\mathbf{0}$ is independent of $i\in\mathcal{V}$.
\end{restatable}


\begin{proof}  
  The proof is based on two inequalities,
  \begin{equation}
  \langle S_{\cal A}S_{\cal B}\rangle_\mathbf{0}\pm  \langle S_{\cal
    A}S_{\cal B}\rangle_\mathbf{e}\ge\left|  \langle S_{\cal
    A}\rangle_\mathbf{0}
  \langle S_{\cal B}\rangle_\mathbf{e}\pm  \langle S_{\cal
    A}\rangle_\mathbf{e}\langle S_{\cal B}\rangle_0\right|,
\label{eq:comparison-ineq}  
\end{equation}
where ${\cal A}\subset{\cal V}$ and ${\cal B}\subset{\cal V}$ are sets
of vertices.  The inequality with the lower (negative) signs is the
Lebowitz comparison inequality\cite{Lebowitz-1977}, while the
inequality with the upper signs can be proved using the same
technique.  In the case of an Ising model on a graph
${\cal G}=({\cal V},{\cal E})$, we have%
  $$
  {d\delta_\mathbf{e}(K)\over dK}=\sum_{ ij=b\in{\cal E}}\left[
  \langle S_iS_j\rangle_\mathbf{0}-  (-1)^{e_b}\langle
  S_iS_j\rangle_\mathbf{e}\right] .
  $$
  Applying Eq.~(\ref{eq:comparison-ineq}) for each term separately,
  with the help of transitivity, $\langle S_i\rangle_\mathbf{0}\equiv
  m_0\ge0$, $i\in{\cal V}$, one gets 
  \begin{equation}
    {d\delta_\mathbf{e}(K)\over dK}\ge m_0\sum_{b=ij\in{\cal E}}
    |m_i'-(-1)^{e_b}m_j'|,\label{eq:delta-derivative}  
  \end{equation}
  where $m_i'\equiv \langle S_i\rangle_\mathbf{e}$.  The statement of
  the Lemma is obtained by noticing that for a path connecting $1$ and
  $f$,
  $$|m_1'-m_2'|+|m_2'-m_3'|+\ldots +|m'_{f-1}-m'_f|\ge m'_f-m_1',$$
  which allows to trade $\wgt({\bf e})$ terms with $+$ signs in the
  r.h.s.\ of Eq.~(\ref{eq:delta-derivative}) for the sum of
  magnetizations $m_v'$ on the $2\wgt(\mathbf{e})$ vertices from
  ${\cal A}$.
\end{proof}

\bibliography{lpp,qc_all,more_qc,percol,spin,grav,sg,htseries,analiz,teach}

\end{document}